\def\eprintversion{1}
\ifnum\eprintversion=0
	\documentclass{llncs}
\else
	\documentclass{llncs}[11pt]
	\usepackage{fullpage}
\fi

\newif\ifsubmission
\submissionfalse

\usepackage[dvips]{graphicx} 
\usepackage{enumerate}
\usepackage{amsmath}
\usepackage{amsfonts}
\usepackage{amssymb}
\usepackage{color} 
\usepackage{multicol}
\usepackage{mathabx}
\usepackage{calc}
\usepackage[
	pdftitle={Quantum Key-length Extension},
    pdfauthor={Joseph Jaeger, Fang Song, Stefano Tessaro},colorlinks=true,linkcolor=blue,citecolor=blue]{hyperref}
\usepackage{colonequals}  
\usepackage{cleveref}
\usepackage{xcolor}
\usepackage{wrapfig}

\definecolor{rvone}{RGB}{139,34,82}
\definecolor{rvtwo}{rgb}{0.00,0.62,0.22}

\renewcommand{\epsilon}{\varepsilon}

\renewcommand{\Pr}{\mathsf{Pr}}

\newcommand{\N}{\mathbb{N}}

\newcommand{\bigoh}{\mathcal{O}}
\newcommand{\Adv}{\mathsf{Adv}}

\newcommand{\sG}{{\normalfont \textsf{G}}}
\newcommand{\sH}{{\normalfont \textsf{H}}}

\newcommand{\bD}{\mathbf{D}}

\newcommand{\advA}{\mathcal{A}}
\newcommand{\advD}{\mathcal{B}}
\newcommand{\distD}{\bD}

\newcommand{\Exp}[1]{\mathsf{E}\left[#1\right]}

\makeatletter
\newcommand\parag{%
   \@startsection{paragraph}{4}{\z@}%
       {-3\p@ \@plus -1\p@ \@minus -1\p@}%
       {-0.2em \@plus -0.12em \@minus -0.05em}%
       {\normalfont\normalsize\scshape}}
\newcommand{\heading}[1]{\parag{#1}}
\makeatother

\newcommand{\ind}{\hspace*{10pt}}

\newcommand{\true}{\texttt{true}}

\newcommand{\concat}{\,\|\,}

\newcommand{\sett}[1]{\{#1\}}
\newcommand{\getsr}{{\:{\leftarrow{\hspace*{-3pt}\raisebox{.75pt}{$\scriptscriptstyle\$$}}}\:}}
\newcommand{\bits}{\{0,1\}}
\newcommand{\xor}{\oplus}

\newcommand{\setI}{\mathcal{I}}
\newcommand{\setO}{\mathcal{O}}

\newcommand{\commentt}[1]{\hspace{2pt}{\small /$\!$/ \textbf{#1}}}
\newcommand{\tH}{\widetilde{H}}

\definecolor{highlight-gray}{gray}{0.90}
\definecolor{highlight-yellow}{cmyk}{0,0,0.90,0}
\definecolor{highlight-cyan}{cmyk}{0.40,0,0,0}
\definecolor{highlight-magenta}{cmyk}{0,0.90,0,0}

\newcommand{\gamechange}[2][highlight-gray]{{\setlength{\fboxsep}{0pt}\colorbox{#1}{\ifmmode$\displaystyle#2$\else#2\fi}}}
\newcommand{\gamechangey}[2][highlight-yellow]{{\setlength{\fboxsep}{0pt}\colorbox{#1}{\ifmmode$\displaystyle#2$\else#2\fi}}}
\newcommand{\gamechangec}[2][highlight-cyan]{{\setlength{\fboxsep}{0pt}\colorbox{#1}{\ifmmode$\displaystyle#2$\else#2\fi}}}
\newcommand{\gamechangem}[2][highlight-magenta]{{\setlength{\fboxsep}{0pt}\colorbox{#1}{\ifmmode$\displaystyle#2$\else#2\fi}}}

\newcommand{\hlg}[1]{\gamechange{#1}}

\newcommand{\gamechanger}[2][highlight-green]{{\setlength{\fboxsep}{0pt}\colorbox{#1}{\ifmmode$\displaystyle#2$\else#2\fi}}}
\colorlet{highlight-green}{green!40!yellow}

\newcommand{\Func}{\mathsf{Fcs}}

\newcommand{\Icm}{\mathsf{Ics}}
\newcommand{\Inj}{\mathsf{Inj}}

\newcommand{\ro}{\mathbf{H}}
\newcommand{\ic}{\mathbf{E}}

\newcommand{\FF}{\mathsf{F}}
\newcommand{\EF}{\mathsf{E}}

\newcommand{\FX}{\mathsf{FX}}
\newcommand{\FFX}{\mathsf{FFX}}
\newcommand{\DE}{\mathsf{DE}}

\newcommand{\Dom}[1]{\mathsf{Dom}(#1)}

\newcommand{\il}[1]{#1.\mathsf{il}}
\newcommand{\ol}[1]{#1.\mathsf{ol}}
\newcommand{\bl}[1]{#1.\mathsf{bl}}
\newcommand{\kl}[1]{#1.\mathsf{kl}}

\newcommand{\prf}{\mathsf{prf}}

\newcommand{\sprp}{\mathsf{sprp}}

\newcommand{\kpa}{\mbox{-}\mathsf{na}}

\newcommand{\dist}{\mathsf{dist}}

\newcommand{\guess}{\mathsf{guess}}
\newcommand{\ldis}{\mathsf{ld}}

\newcommand{\dc}{\mbox{-}\mathsf{d}}
\newcommand{\srch}{\mbox{-}\mathsf{s}}

\newcommand{\figref}[1]{Fig.~\ref{#1}}
\newcommand{\secref}[1]{Section~\ref{#1}}
\newcommand{\apref}[1]{Appendix~\ref{#1}}
\newcommand{\thref}[1]{Theorem~\ref{#1}}
\newcommand{\lemref}[1]{Lemma~\ref{#1}}

\newcommand{\procEv}{\textsc{Ev}}
\newcommand{\procInv}{\textsc{Inv}}
\newcommand{\procRo}{\textsc{Ro}}
\newcommand{\procOrac}{\textsc{O}}
\newcommand{\procIc}{\textsc{Ic}}
\newcommand{\procIcInv}{\textsc{Inv}}

\newcommand{\procFEv}{\textsc{FEv}}
\newcommand{\procFRo}{\textsc{FRo}}

\newcommand{\had}{\mathcal{H}}
\newcommand{\dec}{\mathcal{T}}

\newcommand{\lang}{\mathcal{L}}
\newcommand{\rel}{\mathcal{R}}
\newcommand{\bigi}{\mathcal{I}}

\newcommand{\dom}{D}
\newcommand{\rng}{R}

\newcommand{\pp}{\textsf{PB}}
\newcommand{\ld}{\textsf{1LD}}
\newcommand{\ed}{\textsf{ED}}
\newcommand{\oed}{\textsf{1ED}}

\newcommand{\suchthat}{{\mbox{\ s.t.\ }}}

\newcommand{\gamesfontsize}{\small}

\newcommand{\oneCol}[2]{
\begin{center}
        \framebox{
        \begin{tabular}{c@{\hspace*{.4em}}}
        \begin{minipage}[t]{#1\textwidth}
        	\gamesfontsize #2 \end{minipage}
        \end{tabular}
        }
\end{center}
}

\newcommand{\oneColNoBox}[1]{
\begin{center}\begin{tabular}{c}
\begin{minipage}[t]{1in}\begin{tabbing}
123\=123\=123\=\kill
#1
\end{tabbing}\end{minipage} 
\end{tabular}\end{center}}

\newcommand{\twoColsNoDivide}[4]{
\begin{center}
        \framebox{
        \begin{tabular}{c@{\hspace*{.4em}}c@{\hspace*{.4em}}c}
        \begin{minipage}[t]{#1\textwidth}
        	\gamesfontsize #3 \end{minipage}
        &
        \begin{minipage}[t]{#2\textwidth}
        	\gamesfontsize #4 \end{minipage}
        \end{tabular}
        }
\end{center}
}

\newcommand{\threeColsNoDivide}[6]{
\begin{center}
        \framebox{
        \begin{tabular}{c@{\hspace*{.4em}}c@{\hspace*{.4em}}c}
        \begin{minipage}[t]{#1\textwidth}
        	\gamesfontsize #4 \end{minipage}
        &
        \begin{minipage}[t]{#2\textwidth}
        	\gamesfontsize #5 \end{minipage}
        &
        \begin{minipage}[t]{#3\textwidth}
        	\gamesfontsize #6 \end{minipage}
        \end{tabular}
        }
\end{center}
}

\usepackage{thm-restate}

\usepackage{physics}
\newif\ifnotes
\notesfalse
\ifnotes
\newcommand{\js}[1]{$\ll$\textsf{\color{blue} JS: { #1}}$\gg$}
\newcommand{\st}[1]{$\ll$\textsf{\color{red} ST: { #1}}$\gg$}
\newcommand{\fs}[1]{$\ll$\textsf{\color{cyan} FS: { #1}}$\gg$}

\else
\newcommand{\js}[1]{}
\newcommand{\st}[1]{}
\newcommand{\fs}[1]{}
\fi

\begin{document}

\ifnum\eprintversion=0
	\title{Quantum Key-length Extension}
\else
	\title{Quantum Key-length Extension}
\fi

\ifnum\eprintversion=0
	\ifsubmission
		\author{Joseph Jaeger\inst{1} \and Fang Song\inst{2} \and Stefano Tessaro\inst{1}}
		\institute{
			Paul G.\ Allen School of Computer Science \& Engineering\\
			University of Washington, Seattle, US\\
			\email{\{jsjaeger,tessaro\}@cs.washington.edu}
			\and
			Portland State University, Portland, Oregon, US\\
			\email{fang.song@pdx.edu}}
	\else
		\author{Joseph Jaeger\inst{1} \and Fang Song\inst{2} \and Stefano Tessaro\inst{3}}
		\institute{
			Georgia Institute of Technology, Atlanta, Georgia, US\\
`			\email{josephjaeger@gatech.edu}
			\and
			Portland State University, Portland, Oregon, US\\
			\email{fang.song@pdx.edu}
			\and
			Paul G.\ Allen School of Computer Science \& Engineering\\
			University of Washington, Seattle, Washington, US\\
			\email{tessaro@cs.washington.edu}	
		}
	\fi
\else
		\author{Joseph Jaeger\inst{1} \and Fang Song\inst{2} \and Stefano Tessaro\inst{3}}
		\institute{
			Georgia Institute of Technology, Atlanta, Georgia, US\\
			\email{josephjaeger@gatech.edu}
			\and
			Portland State University, Portland, Oregon, US\\
			\email{fang.song@pdx.edu}
			\and
			Paul G.\ Allen School of Computer Science \& Engineering\\
			University of Washington, Seattle, Washington, US\\
			\email{tessaro@cs.washington.edu}	
			}
\fi

\maketitle

\begin{abstract}

  Should quantum computers become available, they will reduce the
  effective key length of basic secret-key primitives, such as
  blockciphers. To address this we will either need to use
  blockciphers with inherently longer keys or develop key-length
  extension techniques to amplify the security of a
  blockcipher to use longer keys.

  We consider the latter approach and revisit the FX and double encryption
  constructions. 
  Classically, FX was proven to be a secure key-length extension
  technique, while double encryption fails to be more secure than
  single encryption due to a meet-in-the-middle attack. In this work
  we provide positive results, with concrete and tight bounds, for the
  security of \emph{both} of these constructions against quantum
  attackers in ideal models.

  For FX, we consider a partially-quantum model, where the attacker
  has quantum access to the ideal primitive, but only classical access
  to FX. This is a natural model and also the strongest possible,
  since effective quantum attacks against FX exist in the
  fully-quantum model when quantum access is granted to both
  oracles. We provide two results for FX in this model. The first
  establishes the security of FX against non-adaptive attackers. The
  second establishes security against general adaptive attackers for a
  variant of FX using a random oracle in place of an ideal
  cipher. This result relies on the techniques of Zhandry (CRYPTO '19)
  for lazily sampling a quantum random oracle. An extension to
  perfectly lazily sampling a quantum random permutation, which would
  help resolve the adaptive security of standard FX, is an important
  but challenging open question. We introduce techniques for
  partially-quantum proofs without relying on analyzing the classical
  and quantum oracles separately, which is common in existing
  work. This may be of broader interest.
  
  For double encryption, we show that it amplifies strong
  pseudorandom permutation security in the fully-quantum model,
  strengthening a known result in the weaker sense of key-recovery
  security. This is done by adapting a technique of Tessaro and
  Thiruvengadam (TCC '18) to reduce the security to the difficulty of
  solving the list disjointness problem and then showing its hardness
  via a chain of reductions to the known quantum difficulty of the
  element distinctness problem.

\end{abstract}

\section{Introduction}

The looming threat of quantum computers has inspired significant
efforts to design and analyze post-quantum cryptographic schemes. In
the public-key setting, polynomial-time quantum algorithms for
factoring and computing discrete logarithms essentially break all
practically deployed primitives~\cite{FOCS:Shor94}. 

In the secret-key setting, Grover's quantum search
algorithm~\cite{grover} will reduce the effective key length of
secret-key primitives by half. Thus, a primitive like the AES-128
blockcipher which may be thought to have 128 bits of security against
classical computers may provide no more than 64 bits of security
against a quantum computer, which would be considered significantly
lacking.  Even more worrisome, it was shown relatively recent that
quantum computers can break several secret-key constructions
completely such as the Even-Mansour blockcipher~\cite{KM12} and
CBC-MAC~\cite{ToSC:KLLN16} if we grant the attacker fully quantum
access to the cryptosystem.

This would not be the first time that we find ourselves using too
short of a key.  A similar issue had to be addressed when the DES
blockcipher was widely used and its 56 bit keylength was considered
insufficient.  Following approaches considered at that time, we can
either transition to using basic primitives which have longer keys
(e.g. replacing AES-128 with AES-256) or design key-length extension
techniques to address the loss of concrete security due to quantum
computers.  In this paper we analyze the latter approach.  We consider
two key-length extension techniques, FX~\cite{JC:KilRog01} and double
encryption, and provide provable bounds against quantum attackers in
{\em ideal} models.

Of broader and independent interest, our study of FX focuses on a
hybrid quantum model which only allows for {\em classical} online
access to the encrypted data, whereas offline computation is
quantum. This model is sometimes referred to as the ``Q1 model'' in
the cryptanalysis literature~\cite{ToSC:KLLN16,BHNPS19}, in contrast
to the fully-quantum, so-called ``Q2 model'', which allows for quantum
online access. This is necessary in view of existing attacks in the Q2
model showing that FX is no more secure than the underlying
cipher~\cite{AC:LeaMay17}, but {\em also}, Q1 is arguably more
realistic and less controversial than Q2. We observe that (as opposed
to the plain model) ideal-model proofs in the Q1 model can be {\em
  harder} than those in the Q2 model, as we need to explicitly account
for {\em measuring} the online queries to obtain improved bounds. {In
  many prior ideal-model Q1 proofs,
  e.g.~\cite{EC:KilLyuSch18,KSS+20,HKSU20,BHH+19,KYY21}, this
  interaction is handled essentially {for free} because the effects of
  online and offline queries on an attacker's advantage can largely be
  analyzed separately.} Our work introduces techniques to handle the 
  interaction
between classical online queries and quantum offline ideal-model
queries in Q1 proofs that cannot be analyzed separately. On the other hand, our result on double
encryption considers the full Q2 model -- and interestingly,
restricting adversaries to the Q1 model does not improve the bound. To
be self-explanatory we will often refer to the Q1 and Q2 models as the
partially-quantum and fully-quantum models, respectively.

The remainder of this introduction provides a detailed overview of our
results for these two constructions, and of the underlying challenges
and techniques.

\subsection{The FX Construction}
The FX construction was originally introduced by Kilian and Rogaway~\cite{JC:KilRog01} as a generalization of Rivest's DESX construction.
Consider a blockcipher $\EF$ which uses a key $K\in\bits^k$ to encrypt messages $M\in\bits^n$.
Then the FX construction introduces ``whitening'' key $K_2\in\bits^n$ which is xor-ed into the input and output of the blockcipher.
Formally, this construction is defined by $\FX[\EF](K \concat K_2, M)= \EF_{K}(M \xor K_2) \xor K_2$.
(Note that the Even-Mansour blockcipher~\cite{JC:EveMan97} may be
considered to be a special case of this construction where $k=0$,
i.e., the blockcipher is a single permutation.)
This construction has negligible efficiency overhead as compared to using $\EF$ directly.

Kilian and Rogaway proved this scheme secure against classical attacks in the ideal cipher model.
In particular, they established that
\begin{displaymath}
  \Adv^{\sprp}_{\FX}(\advA)\leq pq/2^{k+n-1}\, . 
\end{displaymath}
Here $\Adv^{\sprp}$ measures the advantage of $\advA$ in breaking the strong pseudorandom permutation (SPRP) security of $\FX$ while making at most $p$ queries to the ideal cipher and at most $q$ queries to the $\FX$ construction.
Compared to the $p/2^{k}$ bound achieved by $\EF$ alone, this is a clear improvement so $\FX$ can be considered a successful key-length extension technique again classical attackers. 

Is this construction equally effective in the face of quantum
attackers? The answer is unfortunately negative. Leander and
May~\cite{AC:LeaMay17}, inspired by a quantum attack due to Kuwakado
and Morii~\cite{KM12} that completely breaks Even-Mansour blockcipher,
gave a quantum attack against FX, which shows that the whitening keys
provide essentially no additional security over that achieved by $\EF$
in isolation. Bonnetain, et al.~\cite{BHNPS19} further reduced
the number of \emph{online} quantum queries in the attack. Roughly
speaking, $O(n)$ quantum queries to FX construction and $O(n 2^{k/2})$
local quantum computations of the blockcipher suffice to recover the secret
encryption key.
Note that, however, such attacks require full quantum access to both
the ideal primitive and to the instance FX that is under attack, i.e.
they are attacks in the fully-quantum model. The latter is rather
strong and may be considered unrealistic. While we cannot prevent a
quantum attacker from locally evaluating a blockcipher in quantum
superposition, honest implementations of encryption will likely
continue to be classical.\footnote{Some may argue that maintaining
  purely classical states, e.g., enforcing perfect measurements is
  also non-trivial physically. However, we deem maintaining coherent
  quantum superposition significantly more
  challenging.}

\heading{Partially-quantum model.} Because of the realistic concern of
the fully-quantum model and the attacks therein that void key
extension in FX, we turn to the {partially-quantum} model in which
the attacker makes quantum queries to ideal primitives, but only
\emph{classical} queries to the cryptographic constructions.

In this model there has been extensive quantum cryptanalysis on FX and
related constructions~\cite{HS18,BHNPS19}. The best existing
attack~\cite{BHNPS19} recovers the key of FX using roughly
$2^{(k+n)/3}$ classical queries to the construction and $2^{(k+n)/3}$
quantum queries to the ideal cipher. However, to date, despite the
active development in provable quantum security, we are not aware of
SPRP or just PRF security analysis, which gives stronger security
guarantees. Namely, it should not just be infeasible to retrieve a
key, but also to merely distinguish the system from a truly random
permutation (or function). We note that in the special case where the
primitives are plain-model instantiations (e.g., non-random-oracle
hash functions), with a bit of care many security reductions carry
over to the quantum setting~\cite{PQCRYPTO:MenSze17}. This is because
the underlying primitives are hidden from the adversary, and hence the
difficulty arising from the interaction of classical and quantum
queries to two correlated oracles becomes irrelevant.

Our main contribution on FX is to prove, for the first time,
indistinguishability security in the partially-quantum model, in two
restricted ways. Although they do not establish the complete security,
our security bounds are tight in their respective
settings.\footnote{Throughout, when mentioning tightness, we mean it
  with respect to the resources required to achieve advantage around
  one. The roots in our bounds make them weaker for lower resource
  regimes. Removing these is an interesting future direction.}

\heading{Non-adaptive security.}
We first consider non-adaptive security where we restrict the adversary such that
its classical queries to the FX construction (but not to the underlying ideal cipher) must be specified before execution has begun.
We emphasize that non-adaptive security of a blockcipher suffices to prove adaptive security for many practical uses of blockciphers such as the various randomized or stateful encryption schemes (e.g. those based on counter mode or output feedback mode) in which an attacker would have no control over the inputs to the blockcipher.

In this setting the bound we prove is of the form 
\begin{displaymath}
\Adv^{\sprp\kpa}_{\FX}(\advA)\leq O\left(\sqrt{p^2q/2^{k+n}}\right).
\end{displaymath}
Supposing $k=n=128$ (as with AES-128), an attacker able to make
$p\approx 2^{64}$ queries to the ideal cipher could break security of
$\EF$ in isolation.  But to attack FX, such an attacker with access to
$q\approx 2^{64}$ encryptions would need to make $p\approx 2^{96}$
queries to the ideal cipher. In fact we can see from our bound that
breaking the security with constant probability would require the
order of $\Omega(2^{(k+n)/3})$ queries in total, matching the bound
given in the attacks mentioned above~\cite{BHNPS19}. Hence our bound
is tight.

To prove this bound we apply a one-way to hiding (O2H) theorem of
Ambainis, Hamburg, and Unruh~\cite{C:AmbHamUnr19}, an improved version
of the original one in~\cite{unruh}. This result provides a clean
methodology for bounding the probability that an attacker can
distinguish between two functions drawn from closely related
distributions
given quantum access.  The non-adaptive setting allows us to apply
this result by sampling the outputs of FX ahead of time and then
considering the ideal world in which the ideal cipher is chosen
independently of these outputs and the real world in which we very
carefully reprogram this ideal cipher to be consistent with the
outputs chosen for FX.  These two ideal ciphers differ only in the
$O(q)$ places where we need to reprogram.

\heading{Adaptive security of FFX.}
As a second approach towards understanding fully adaptive security of
FX, we consider a variant construction (which we call FFX for
``function FX'') that replaces the random permutation with a random
function.  In particular, suppose $\FF$ is a function family which uses a
key $K\in\bits^k$ on input messages $M\in\bits^n$ to produce outputs
$C\in\bits^m$. Then we define
$\FFX[\FF](K\concat K_2, M)= \FF_{K}(M \xor K_2)$.\footnote{Note we
  have removed the external xor with $K_2$. In FX this xor is necessary, but in our analysis it would not provide any benefit for FFX.}  For this
construction we prove a bound of the form
\begin{displaymath}
	\Adv^{\prf}_{\FFX}(\advA)\leq O\left(\sqrt{{p^2q}/{2^{k+n}}}\right).
\end{displaymath}
in the partially-quantum random oracle model. Note that this matches
the bound we obtained for the non-adaptive security of FX. Since the
same key-recovery attack~\cite{BHNPS19} also applies here, it follows
that our bound is tight as well. Our proof combines two techniques of
analyzing a quantum random oracle, the O2H theorem above and a
simulation technique by Zhandry~\cite{C:Zhandry19}. The two techniques
usually serve distinct purposes. O2H is helpful to program a random
oracle, whereas Zhandry's technique is typically convenient for
(compactly) maintaining a random oracle and providing some notion of
``recording'' the queries. In essence, in the two function
distributions of O2H for which we aim to argue indistinguishability,
we apply Zhandry's technique to simulate the functions in a compact
representation. As a result, analyzing the guessing game in O2H, which
implies indistinguishability, becomes intuitive and much
simplified. This way of combining them could also be useful elsewhere.

To build intuition for the approach of our proof, let us first consider one way to prove the security of this construction classically.
The core idea is to use lazy sampling.
In the ideal world, we can independently lazily sample a random
function $F:\bits^k\cross\bits^n\to\bits^m$ to respond to $\FF$
queries and a random function $T:\bits^n\to\bits^m$ to respond to
$\FFX$ queries. 
These lazily random functions  are stored in tables.

The real world can similarly be modeled by lazily sampling $F$ and $T$ to respond to the separate oracles.
However, these oracles need to be kept consistent.
So if the adversary ever queries $M$ to $\FFX$ and $(K, M\xor K_2)$ to $\FF$, then the game should copy values between the two tables such that the same value is returned by both oracles.
(Here $K$ and $K_2$ are the keys honestly sampled by the game.)
Alternatively, we can think of the return value being stored only in the $T$ table when such queries occur (rather than being copied into both tables) as long as we remember that this has happened.
When represented in this manner, the two games only differ if the adversary makes such a pair of queries, where we think of the latter one as being ``bad''.
Thus a simple $O(pq/2^{k+n})$ bound on the probability of making such a query bounds the  advantage of the adversary.

In our quantum security proof we wish to proceed analogously.
First, we find a way to represent the responses to oracle queries with two (superpositions over) tables that are independent in the ideal world and dependent in the real world (the dependency occurs only for particular ``bad'' inputs).
Then (using the O2H theorem of Ambainis, Hamburg, and Unruh) we can bound the distinguishing advantage by the probability of an attacker finding a ``bad'' input.
In applying this theorem we will jointly think of the security game and its adversary $\advA$ as a combined adversary $\advA'$ making queries to an oracle which takes in both the input of $\advA$ \emph{and} the tables being stored by the game -- processing them appropriately.

The required representation of the oracles via two tables is a highly
non-trivial step in the quantum setting.  For starters, the no-cloning
theorem prevents us from simply recording queries made by the
adversary.  This has been a recurring source of difficulty for
numerous prior papers such as~\cite{FOCS:Zhandry12,TCC:TarUnr16,AC:DagFisGag13}.  We make use of
the recent elegant techniques of Zhandry~\cite{C:Zhandry19} which
established that, by changing the perspective (e.g., to the Fourier
domain), a random function can be represented by a table which is
initialized to all zeros and then xor-ed into with each oracle query
made by the adversary. This makes it straightforward to represent the
ideal world as two separate tables. To represent the real world
similarly, we exploit the fact that the queries to FFX are classical.
To check if an input to FFX is ``bad'' we simply check if the corresponding
entry of the random oracle's table is non-zero. To check if an input to the 
random oracle is ``bad'' we check if it overlaps with prior queries to FFX
which we were able to record because they were classical.
For a ``bad'' input we then share the storage of the two tables and
this is the only case where the behavior of the real world differs from that
of the ideal world. These ``bad'' inputs may of course be part of a superposition
query and it is only for the bad components of the superposition that the 
games differ.

\heading{Difficulty of extending to FX.}  It is possible that this
proof could be extended to work for normal FX given an analogous way
to lazily represent a random permutation.  Unfortunately, no such
representation is known.

Czajkowski, et al.~\cite{EPRINT:CMSZ19} extended Zhandry's lazy
sampling technique to a more general class of random functions, but
this does not include permutations because of the correlation between
the different outputs of a random permutation.  Chevalier, et
al.~\cite{CEV20} provided a framework for recording
queries to quantum oracles, which enables succinctly recording queries
made to an externally provided function for purposes of later
responding to inverse queries. This is distinct from the lazy
sampling of a permutation that we require.
Rosmanis~\cite{rosmanis2021tight} introduced a new technique for
analyzing random permutations in a compressed manner and applied it to
the question of inverting a permutation (given only forward access to
it).
Additional ideas seem needed to support actions based on the oracle
queries that have been performed so far. This is essential in order to
extend our proof for the function variant of FX to maintain
consistency for the real world in the face of ``bad'' queries.

Recent work of Czajkowski~\cite{EPRINT:Czajkowski21} provided an \emph{imperfect} lazy sampling technique for permutations and used it to prove indifferentiability of SHA3. They claim that their lazy sampling strategy cannot be distinguished from a random permutation with advantage better than $O(q^{2} / 2^n)$. Unfortunately, this bound is too weak to be useful to our FX proof. For example, if $k\geq n$ we already have $O(q^2 / 2^n)$ security without key-length extension.
Determining if it is possible to \emph{perfectly} lazily sample a
random permutation remains an interesting future direction.

\subsection{Double Encryption}
The other key-extension technique we consider is double encryption.
Given a blockcipher $\EF:\bits^k\cross\bits^n\to\bits^n$ this is
defined by $\DE[\EF](K_1\concat K_2, M)= \EF_{K_2}(\EF_{K_1}(M))$.
This construction requires more computational overhead than FX
because it requires two separate application of the blockcipher with
different keys.  Classically, this construction is not considered to
be a successful key-length extension technique because the
meet-in-the-middle
attack~\cite{diffie1977exhaustive,merkle1981security} shows that it
can be broken in essentially the same amount of time as $\EF$ alone.

However, this does not rule out that double encryption is
an effective key-length extension method in the quantum setting, as
it is not clear that the Grover search algorithm~\cite{grover} used to
halve the effective keylength of blockciphers can be composed with the
meet-in-the middle attack to unify their savings.  The security of
double encryption in the quantum setting was previously considered by
Kaplan~\cite{kaplan2014quantum}. They related the key-recovery problem
in double encryption to the claw-finding problem, and gave the tight
quantum query bound $\Theta(N^{2/3})$ for solving key recovery (here
$N=2^k$ is the length of the lists in the claw-finding problem). This
indicates that in the quantum setting double encryption is in fact
useful (compare to $N^{1/2}$), although key-recovery security is
fairly weak.

We strengthen their security result by proving the SPRP security,
further confirming double encryption as an effective key-extension
scheme against quantum attacks. This is proven in the
\emph{fully-quantum} model, and the bound we obtain matches the attack
in~\cite{kaplan2014quantum} which works in the partially-quantum
model. Namely restricting to the weaker partially-quantum model would
not improve the bound. Our result is obtained by a reduction to list
disjointness. This is a worst-case decision problem measuring how well
an algorithm can distinguish between a pair of lists with zero or
\emph{exactly} one element in common, which can be viewed as a
decision version of the claw-finding problem. This reduction technique
was originally used by Tessaro and Thiruvengadam~\cite{TCC:TesThi18}
to establish a classical time-memory trade-off for double
encryption. We observe that their technique works for a quantum
adversary.

We then construct a chain of reductions to show that the
known quantum hardness of element distinctness~\cite{AS04,Zhandry15}
(deciding if a list of $N$ elements are all distinct) can be used to
establish the quantum hardness of solving list disjointness.  Our
result (ignoring log factors) implies that a highly successful
attacker must make $\Omega(2^{2k/3})$ oracle queries which is more
than the $\Omega(2^{k/2})$ queries needed to attack $\EF$ used in
isolation.

Our proof starts by observing that Zhandry's~\cite{Zhandry15} proof of
the hardness of the search version of element distinctness (finding a
collision in a list) in fact implies that a promise version of element
distinctness (promising that there is exactly one collision) is also
hard. Then a simple reduction (randomly splitting the element
distinctness list into two lists) shows the hardness of the search
version of list disjointness. Next we provide a binary-search inspired
algorithm showing that the decision version of list disjointness can
be used to solve the search version, implying that the decision
version must be hard.  During our binary search we pad the lists we
are considering with random elements to ensure that our lists maintain
a fixed size which is necessary for our proof to go through.

The final bound we obtain for double encryption is of the form
\begin{displaymath}
		\Adv^{\sprp}_{\DE}(\advA)\leq O\left(\sqrt[6]{(q\cdot k\lg k)^3/2^{2k}}\right).
\end{displaymath}
The sixth root arises in this bound from the final step in our chain of results analyzing list disjointness.
The binary search algorithm requires its underlying decision list disjointness algorithm to have relatively high advantage.
To obtain this from a given algorithm with advantage $\delta$ we need to amplify its advantage by running in on the order of $1/\delta^2$ times.
The number of queries depending on the square of
$\delta$ causes the root to arise in the proof. 

\subsection{Overview}
In \secref{sec:prelims}, we introduce preliminaries such as notation, basic cryptographic definitions, and some background on quantum computation that we will use throughout the paper.
Following this, in \secref{sec:fx} we consider the security of FX in the partially quantum setting.
Non-adaptive SPRP security of FX is proven in \secref{sec:fx-kpa} and adaptive PRF security of {FFX} is proven in \secref{sec:ffx-adaptive}.
We conclude with \secref{sec:double-enc} in which we prove the SPRP security of double encryption against fully quantum adaptive attacks.


\section{Preliminaries}\label{sec:prelims}
For $n,m\in\N$, we let $[n]=\sett{1,\dots,n}$ and $[n..m]=\sett{n,n+1,\dots,m}$.
The set of length $n$ bit strings is denoted $\bits^n$.
We use $\concat$ to denote string concatenation.
We let $\Inj(n,m)$ denote the set of injections $f:[n]\to[m]$.

We let $y\getsr\advA[O_1,\dots](x_1,\dots)$ denote the (randomized) execution of algorithm $\advA$ with input $x_1,\dots$ and oracle access to $O_1,\dots$ which produces output $y$.
For different $\advA$ we will specify whether it can access its oracles in quantum superposition or only classically.
If $\mathcal{S}$ is a set, then $y\getsr\mathcal{S}$ denotes randomly sampling $y$ from $\mathcal{S}$.

We express security notions via pseudocode games.
See \figref{fig:ex-game} for some example games.
In the definition of games, oracles will sometimes be specified by pseudocode with the following form.
\oneColNoBox{
\underline{Oracle $\procOrac(X_1,\dots : Z_1,\dots)$}\\
//Code defining $X'_1,\dots$ and $Z'_1,\dots$\\
Return $(X'_1,\dots : Z'_1,\dots)$\smallskip
}
This notation indicates that $X_1,\dots$ are variables controlled by the adversary prior to the oracle query and $Z_1,\dots$ are variables controlled by the game itself which the adversary cannot access.
At the end of the execution of the oracle, these variables are overwritten with the values indicated in the return statement. 
Looking ahead, we will be focusing on quantum computation so this notation will be useful to make it explicit that $\procOrac$ can be interpreted as a unitary acting on the registers $X_1,\dots$ and $Z_1,\dots$ (because $\procOrac$ will be an efficiently computable and invertible permutation over these values).
If $\ro$ is a function stored by the game, then oracle access to $\ro$ represents access to the oracle that on input $(X,Y:\ro)$ returns $(X,\ro(X)\xor Y : \ro)$.

We define games as outputting boolean values and let $\Pr[\sG]$ denote the probability that game $\sG$ returns $\true$.
When not otherwise indicated, variables are implicitly initialized to store all 0's.

If $\advA$ is an adversary expecting access to multiple oracles we say
that it is \emph{order consistent} if the order it will alternate
between queries to these different oracles is a priori fixed before
execution.  Note that order consistency is immediate if, e.g., $\advA$
is represented by a circuit where each oracle is modeled by a separate
oracle gate, but is not immediate for other possible representations
of an adversary.

\heading{Ideal Models.}
In this work we will work in ideal models -- specifically, the random oracle model or the ideal cipher model.
Fix $k,n,m\in\N$ (throughout this paper we will treat these parameters as having been fixed already). We let $\Func(k,n,m)$ be the set of all functions $\ro:\bits^k\times\bits^n\to\bits^m$ and $\Icm(k,n)\subset\Func(k,n,n)$ be the set of all functions $\ic:\bits^k\times\bits^n\to\bits^n$ such that $\ic(K,\cdot)$ is a permutation on $\bits^n$.
When convenient, we will write $\ro_K(x)$ in place of $\ro(K, x)$ for $\ro\in\Func(k,n,m)$.
Similarly, we will write $\ic_K(x)$ for $\ic(K,x)$ and $\ic^{-1}_K(\cdot)$ for the inverse of $\ic_K(\cdot)$ when $\ic\in\Icm(k,n)$.
When $K=\varepsilon$ we omit the subscript to $\ro$ or $\ic$.

In the random oracle model, honest algorithms and the adversary are given oracle access to a randomly chosen $\ro\in\Func(k,n,m)$.
In the ideal cipher model, they are given oracle access to $\ic$ and $\ic^{-1}$ for $\ic$ chosen at random from $\Icm(k,n)$.
We refer to queries to these oracles as \emph{primitive queries} and queries to all other oracles as \emph{construction queries}.

\ifnum\eprintversion=0
        \newcommand{\lwid}{0.3}
        \newcommand{\mwid}{}
        \newcommand{\rwid}{0.33}
\else
        \newcommand{\lwid}{0.22}
        \newcommand{\mwid}{0.21}
        \newcommand{\rwid}{0.25}
\fi
\begin{figure}[t]
\twoColsNoDivide{\lwid}{\rwid}
{
\underline{Game $\sG^{\prf}_{\FF,b}(\advA)$}\\
$\ro\getsr\Func(k,n,m)$\\
$K\getsr\bits^{\kl{\FF}}$\\
$F\getsr\Func(0,\il{\FF},\ol{\FF})$\\
$b'\getsr\advA[{\procEv,\ro}]$\\
Return $b'=1$\smallskip
}
{
\underline{$\procEv(X, Y : \ro,K,F)$}\\
$Y_1\gets\FF[{\ro}](K,X)$\\
$Y_0\gets F(X)$\\
Return $(X, Y_b\xor Y : \ro,K,F)$\smallskip
}
\twoColsNoDivide{\lwid}{\rwid}
{
\underline{Game $\sG^{\sprp}_{\EF,b}(\advA)$}\\
$\ic\getsr\Icm(k,n)$\\
$K\getsr\bits^{\kl{\EF}}$\\
$P\getsr\Icm(0,\bl{\EF})$\\
$b'\getsr\advA[{\procEv,\procInv,\ic,\ic^{-1}}]$\\
Return $b'=1$\smallskip
}
{
\underline{$\procEv(X,Y:\ic,K,P)$}\\
$Y_1\gets\EF[{\ic}](K,X)$\\
$Y_0\gets P(X)$\\
Return $(X,Y_b\xor Y : \ic,K,P)$\\[4pt]
\underline{$\procInv(X,Y:\ic,K,P)$}\\
$Y_1\gets\EF^{-1}[\ic](K,X)$\\
$Y_0\gets P^{-1}(X)$\\
Return $(X,Y_b\xor Y: \ic,K,P)$\smallskip
}
\vspace{-2ex}
\caption{Security games measuring PRF security of a family of functions $\FF$ and SPRP security of a blockcipher $\EF$.}
\label{fig:ex-game}
\label{fig:prf}\label{fig:prp}\label{fig:sprp}
\hrulefill
\end{figure}

\heading{Function family and pseudorandomness.}
A function family $\FF$ is an efficiently computable element of $\Func(\kl{\FF},\il{\FF},\ol{\FF})$.
If, furthermore, $\FF\in\Icm(\kl{\FF},\il{\FF})$ and $\FF^{-1}$ is efficiently computable then we say $\FF$ is a blockcipher and let $\bl{\FF}=\il{\FF}$.

If $\FF$ is a function family (constructed using oracle access to a
function $\ro\in\Func(k,n,m)$), then its security (in the random
oracle model) as a \emph{pseudorandom function} (PRF) is measured by
the game $\sG^{\prf}$ shown in \figref{fig:prf}.  In it, the adversary
$\advA$ attempts to distinguish between a \emph{real} world ($b=1$)
where it is given oracle access to $\FF$ with a random key $K$ and an
\emph{ideal} world ($b=0$) where it is given access to a random
function.  We define the advantage function
$\Adv^{\prf}_\FF(\advA)=\Pr[\sG^{\prf}_{\FF,1}(\advA)]-\Pr[\sG^{\prf}_{\FF,0}(\advA)]$.

If $\EF$ is a blockcipher (constructed using oracle access to a function $\ic\in\Icm(k,n)$ and its inverse), then its security (in the ideal cipher model) as a \emph{strong pseudorandom permutation} (SPRP) is measured by the game $\sG^{\sprp}$ shown in \figref{fig:sprp}.
In it, the adversary $\advA$ attempts to distinguish between a real world ($b=1$) where it is given oracle access to $\EF$, $\EF^{-1}$ with a random key $K$ and an ideal world ($b=0$) where it is given access to a random permutation.
We define the advantage function $\Adv^{\sprp}_\FF(\advA)=\Pr[\sG^{\sprp}_{\FF,1}(\advA)]-\Pr[\sG^{\sprp}_{\FF,0}(\advA)]$.

In some examples, we will restrict attention to \emph{non-adaptive} SPRP security.
In such cases our attention is restricted to attackers whose queries to $\procEv$ and $\procInv$ when relevant are a priori fixed before execution.
That is, $\advA$ is a non-adaptive attacker which makes at most $q$ classical, non-adaptive queries to $\procEv,\procInv$ if there exists $M_1,\dots,M_{q'},Y_{q'+1},\dots,Y_q\in\bits^n$ such that $\advA$ only ever queries $\procEv$ on $M_i$ for $1\leq i\leq q'$ and $\procInv$ on $Y_i$ for $q'+1\leq i \leq q$.
Then we write $\Adv^{\sprp\kpa}(\advA)$ in place of $\Adv^{\sprp}(\advA)$.

\subsection{Quantum Background}\label{sec:quant-back}
We assume the reader has basic familiarity with quantum computation.
Quantum computation proceeds by performing unitary operations on registers which each contain a fixed number of qubits.
We sometimes use $\circ$ to denote composition of unitaries.
Additionally, qubits may be measured in the computational basis.
We will typically use the principle of deferred measurements to without loss of generality think of such measurements as being deferred until the end of computation.

The Hadamard transform $\had$ acts on a bitstring $x\in\bits^n$ (for some $n\in\N$) via $\had\ket{x} = 1/\sqrt{2^n} \cdot \sum_{x'} (-1)^{x\cdot x'} \ket{x'}$.
Here $\cdot$ denotes inner product modulo 2 and the summation is over $x'\in\bits^n$.
The Hadamard transform is its own inverse.
We sometimes use the notation $\had^{X_1,X_2,\dots}$ to denote the Hadamard transform applied to registers $X_1,X_2,\dots$.

We make use of the fact that if $P$ is a permutation for which both $P$ and $P^{-1}$ can be efficiently implemented classically, then there is a comparable efficient quantumly computable unitary $U_P$ which maps according to $U_P\ket{x} = \ket{P(x)}$ for $x\in\bits^n$.
For simplicity, we often write $P$ in place of $U_P$.
If $f:\bits^n\to\bits^m$ is a function, we define the permutation $f[{\xor}](x,y)=(x,f(x)\xor y)$.

\ifnum\eprintversion=0
        \renewcommand{\lwid}{0.4}
        \renewcommand{\mwid}{}
        \renewcommand{\rwid}{0.45}
\else
        \renewcommand{\lwid}{0.26}
        \renewcommand{\mwid}{}
        \renewcommand{\rwid}{0.29}
\fi
\begin{figure}[t]
\twoColsNoDivide{\lwid}{\rwid}
{
\underline{Game $\sG^{\dist}_{\distD,b}(\advA)$}\\
$(S,S',P_0,P'_0,P_1,P'_1,z)\getsr\distD$\\
$b'\getsr\advA[{P_b},{P'_b}](z)$\\
Return $b'=1$\smallskip
}
{
\underline{Game $\sG^{\guess}_{\distD}(\advA)$}\\
$(S,S',P_0,P'_0,P_1,P'_1,z)\getsr\distD$\\
$i\getsr\sett{1,\dots,q}$\\
Run $\advA[P_0,P'_0]$ until its $i$-th query\\
Measure the input $x$ to this query\\
If the query is to $P_0$ then\\
\ind Return $x\in S$\\
Else (the query is to $P'_0$)\\
\ind Return $x\in S'$\smallskip
}
\vspace{-2ex}
\caption{Games used for O2H \thref{thm:reg-ahu}.}
\label{fig:reg-ahu}
\hrulefill
\end{figure}

\heading{One-way to hiding.}
We will make use of (a slight variant of) a one-way to hiding (O2H) theorem of Ambainis, Hamburg, and Unruh~\cite{C:AmbHamUnr19}.
The theorem will consider an adversary given oracle access either to permutations $(P_0,P'_0)$ or permutations $(P_1,P'_1)$.
It relates the advantage of the adversary in distinguishing between these two cases to the probability that the adversary can be used to find one of those points on which $P_0$ differs from $P_1$ or $P'_0$ differs from $P'_1$.
The result considers a distribution $\distD$ over $(S,S',P_0,P'_0,P_1,P'_1,z)$ where $S,S'$ are sets, $P_0,P_1$ are permutations on the same domain, $P'_0,P'_1$ are permutations on the same domain, and $z\in\bits^\ast$ is some auxiliary information.
Such a $\distD$ is \emph{valid} if $P_0(x)=P_1(x)$ for all $x\not\in S$ and $P'_0(x)=P'_1(x)$ for all $x\not\in S'$.
Now consider the game $\sG^{\dist}_{\distD,b}$ shown in \figref{fig:reg-ahu}.
In it, an adversary $\advA$ is given $z$ and tries to determine which of the oracle pairs it has access to.
We define $\Adv^{\dist}_{\distD}(\advA)=\Pr[\sG^{\dist}_{\distD,1}(\advA)]-\Pr[\sG^{\dist}_{\distD,0}(\advA)]$.

The game $\sG^{\guess}_{\distD}(\advA)$ in the same figure measures
the ability of $\advA$ to query its oracles on inputs at which $P_0$
and $P_1$ (or $P'_0$ and $P'_1$) differ.  It assumes
that the adversary makes at most $q$ oracle queries.  The adversary is
halted in its execution on making a random one of these queries and
the input to this query is measured.  If the input falls in the
appropriate set $S$ or $S'$, then the game returns $\true$.  Thus
we can roughly think of this as a game in which $\advA$ is trying to guess
a point on which the two oracles differ.  We define
$\Adv^{\guess}_{\distD}(\advA)=\Pr[\sG^{\guess}_{\distD}(\advA)]$,
which leads to a bound on $\Adv^{\dist}_{\distD}(\advA)$.

\begin{theorem}[\cite{C:AmbHamUnr19}, Thm.3]\label{thm:reg-ahu}
	Let $\distD$ be a valid distribution and $\advA$ be an adversary making at most $q$ oracle queries.
	Then $\Adv^{\dist}_{\distD}(\advA)\leq 2q\sqrt{\Adv^{\guess}_{\distD}(\advA)}$.
\end{theorem}

Our statement of the theorem differs from the result as given in~\cite{C:AmbHamUnr19} in that we consider arbitrary permutations, rather than permutations of the form $f[{\xor}]$ for some function $f$, and we provide the attacker with access to two oracles rather than one.\footnote{Their result additionally allows the adversary to make oracle queries in parallel and bounds its advantage in terms of the ``depth'' of its oracle queries rather than the total number of queries. We omit this for simplicity.}
These are simply notational conveniences to match how we will be applying the theorem.
The proof given in~\cite{C:AmbHamUnr19} suffices to establish this variant without requiring any meaningful modifications.

The most natural applications of this theorem would apply it to distributions $\distD$ for which the guessing advantage $\Adv^{\guess}_{\distD}(\advA)$ is small for \emph{any} efficient adversary $\advA$.
This will indeed be the case for our use of it in our \thref{thm:fx-kpa}.
However, note that it can also be applied more broadly with a distribution $\distD$ where it is not necessarily difficult to guess inputs on which the oracles differ.
We will do so at the end of our proof of \thref{thm:ffx}.
Here we will use a \emph{deterministic} $\distD$ so, in particular, the sets $S$ and $S'$ are a priori fixed and not hard to query.
The trick we will use to profitably apply the O2H result is to exploit knowledge of the particular form that $\advA$ will take (it will be a reduction adversary internally simulating the view of another adversary) to provide a useful bound on its guessing advantage $\Adv^{\guess}_{\distD}(\advA)$.

\section{The FX Construction}
\label{sec:fx}

The FX construction (originally introduced by Kilian and Rogaway~\cite{JC:KilRog01} as a generalization of Rivest's DESX construction) is a keylength extension for blockciphers.
In this construction, an additional key is used which is xor-ed with input and the output of the blockcipher.\footnote{Technically, the original definition of FX~\cite{JC:KilRog01} uses distinct keys for xor-ing with the input and the output, but this would not provide any benefit in our concrete security analysis so we focus on the simplified construction.}
Formally, given a blockcipher $\EF\in\Icm(\kl{\EF},\bl{\EF})$, the blockcipher $\FX[\EF]$ is defined by $\FX[\EF](K_1\concat K_2, x)=\allowbreak \EF_{K_1}(x \xor K_2) \xor K_2$.
Here $|K_1|=\kl{\EF}$ and $|K_2|=\bl{\EF}$ so $\kl{\FX[\EF]}=\kl{\EF}+\bl{\EF}$ and $\bl{\FX[\EF]}=\bl{\EF}$.
Its inverse can similarly be computed as $\FX[\EF]^{-1}(K_1\concat K_2, x)= \EF^{-1}_{K_1}(x \xor K_2) \xor K_2$.
Let $k=\kl{\EF}$ and $n=\bl{\EF}$.

Kilian and Rogaway~\cite{C:KilRog96} analyzed the PRP security of FX against classical attacks, showing that $\Adv^{\sprp}_{\FX}(\advA) \leq 2pq/2^{k+n}$ where $q$ is the number of $\procEv,\procInv$ queries and $p$ is the number of $\ic,\ic^{-1}$ queries made by $\advA$ (with $\EF$ modeled as an ideal cipher).
In~\cite{AC:LeaMay17}, Leander and May showed a quantum attack against the FX construction -- establishing that the added whitening keys did not provide additionally security.
This attack uses a clever combination of the quantum algorithms of Grover~\cite{grover} and Simon~\cite{simon1997power}.
It was inspired by an attack by Kuwakado and Morii~\cite{KM12} showing that the Even-Mansour blockcipher~\cite{JC:EveMan97} provides no quantum security.
Thus, it seems that $\FX[\EF]$ does not provide meaningfully more security than $\EF$ against quantum attackers.

However, the attack of Leander and May requires quantum access to both the FX construction and the underlying blockcipher $\EF$.
This raises the question of whether the FX is actually an effective key-length extension technique in the partially-quantum setting where the adversary performs only classical queries to the construction oracles.
In this section, we approach this question from two directions.
First, in \secref{sec:fx-kpa} we apply \thref{thm:reg-ahu} with a careful representation of the real and ideal worlds to show that FX does indeed achieve improved security against non-adaptive attacks.

Analyzing the full adaptive security of FX against classical
construction queries seems beyond the capabilities of current proof
techniques.  Accordingly, in \secref{sec:ffx-adaptive}, we consider a
variant of FX in which a random oracle is used in place of the ideal
cipher and prove its quantum PRF security.  Here we apply a new
reduction technique (built on the ``sparse'' quantum representation of
a random function introduced by Zhandry~\cite{C:Zhandry19} and
\thref{thm:reg-ahu}, the O2H theorem from Ambainis, Hamburg, and
Unruh~\cite{C:AmbHamUnr19}) to prove that this serves as an effective
key-length extension technique in our setting.  It seems likely that
our technique could be extended to the normal FX construction, should
an appropriate sparse quantum representation of random permutations be
discovered.

\subsection{Security of FX Against Non-Adaptive Attacks}\label{sec:fx-kpa}

The following theorem bounds the security of the FX construction against non-adaptive attacks (in which the non-adaptive queries are all classical).
This result is proven via a careful use of \thref{thm:reg-ahu} in which the distribution $\distD$ is defined in terms of the non-adaptive queries that the adversary will make and defined so as to perfectly match the two worlds that $\advA$ is attempting to distinguish between. 
\begin{theorem}\label{thm:fx-kpa}
	Let $\advA$ be a quantum adversary which makes at most $q$ classical, non-adaptive queries to $\procEv,\procInv$ and consider $\FX[\cdot]$ with the underlying blockcipher modeled by an ideal cipher drawn from $\Icm(k,n)$.
	Then 
	\begin{displaymath}
		\Adv^{\sprp\kpa}_{\FX}(\advA)\leq \sqrt{8p^2q/2^{k+n}},
	\end{displaymath}
	where $p$ is the number of quantum oracle queries that $\advA$ makes to the ideal cipher. 
\end{theorem}

\begin{proof}
	We will use \thref{thm:reg-ahu}	to prove this result, so first we define a distribution $\distD$.
	Suppose that $M_1,\dots,M_{q'}\in\bits^n$ are the distinct queries $\advA$ will make to $\procEv$ and $Y_{q'+1},\dots,Y_q\in\bits^n$ are the distinct queries that $\advA$ will make to $\procInv$.
	The order in which these queries will be made does not matter.
	Then we define $\distD$ as shown in \figref{fig:fx-dist}.
	This distribution is valid (as required for Theorem~\ref{thm:reg-ahu}) because $G_1$ is reprogrammed to differ from $G_0$ by making inputs in $S$ map to different values in $S'$.
	
	\ifnum\eprintversion=0
        \renewcommand{\lwid}{}
        \renewcommand{\mwid}{}
        \renewcommand{\rwid}{0.9}
	\else
        \renewcommand{\lwid}{}
        \renewcommand{\mwid}{}
        \renewcommand{\rwid}{0.63}
	\fi
	\begin{figure}[t]
	\oneCol{\rwid}
	{
		\underline{Distribution $\distD$}\\
		\commentt{Step 1: Sample responses to construction queries}\\
		For $i=1,\dots,q'$ do\\
		\ind $Y_i\getsr\bits^n\setminus\sett{Y_1,\dots,Y_{i-1}}$\\
		\ind $T[M_i]\gets Y_i$; $T^{-1}[Y_i]\gets M_i$\\
		For $i=q'+1,\dots,q$ do\\
		\ind If $T^{-1}[Y_i]\neq \bot$ then $M_i\gets T^{-1}[Y_i]$\\
		\ind Else $M_i\getsr\bits^n\setminus\sett{M_1,\dots,M_{i-1}}$\\
		\ind $T[M_i]\gets Y_i$; $T^{-1}[Y_i]\gets M_i$\\
		$z\gets (T,T^{-1})$\\
		\commentt{Step 2: Sample $f_0$ as independent ideal cipher}\\
		$f_0\getsr\Icm(k,n)$\\
		\commentt{Step 3: Reprogram $f_1$ for consistency with construction queries}\\
		$K_1\getsr\bits^k$; $K_2\getsr\bits^n$\\
		$\setI\gets\sett{M_i\xor K_2 : 1\leq i \leq q}$;
		$\setO\gets\sett{Y_i\xor K_2 : 1\leq i \leq q}$\\
		$\setI'\gets\sett{f_0^{-1}(K_1,y) : y\in\setO}$;
		$\setO'\gets\sett{f_0(K_1,x) : x\in\setI}$\\
		$S=\sett{(K_1,x):x\in\setI\cup\setI'}$;
		$S'=\sett{(K_1,y):y\in\setO\cup\setO'}$\\
		For $(K,x)\not\in S$ do\\
		\ind $f_1(K,x)\gets f_0(K,x)$\\
		For $i=1,\dots,q$ do\\
		\ind $f_1(K_1,M_i\xor K_2)\gets Y_i \xor K_2$\\
		For $x\in\setI'\setminus\setI$ do\\
		\ind $f_1(K_1,x)\getsr \setO'\setminus\sett{f_1(K_1,x) : x \in \setI\cup\setI', f_1(K_1,x)\neq \bot}$\\
		Return $(S,S',f_0[{\xor}],f_0^{-1}[{\xor}],f_1[{\xor}],f_1^{-1}[{\xor}],z)$\smallskip
	}
	\vspace{-2ex}
	\caption{Distribution of oracles used in proof of \thref{thm:fx-kpa}.}
	\label{fig:fx-dist}
	\hrulefill
	\end{figure}
	
	We will show that the oracles output by this distribution (described in words momentarily) can be used to perfectly simulate the views expected by $\advA$.
	In particular, let $\advA'$ be an adversary (for $\sG^{\dist}_{\distD}$) which runs $\advA$, responding to $\procEv(M_i)$ queries with $T[M_i]$,
	responding to $\procInv(Y_i)$ queries with $T^{-1}[Y_i]$,
	and simulating $\ic,\ic^{-1}$ with its own oracles $f_b[{\xor}],f_b^{-1}[{\xor}]$.
	When $\advA$ halts with output $b$, this adversary halts with the same output.
	We claim that (i) $\Pr[\sG^{\sprp}_{\FF,1}(\advA)]=\Pr[\sG^{\dist}_{\distD,1}(\advA')]$ and (ii) $\Pr[\sG^{\sprp}_{\FF,0}(\advA)]=\Pr[\sG^{\dist}_{\distD,0}(\advA')]$.
	This gives $\Adv^{\sprp\kpa}_\FF(\advA)=\Adv^{\dist}_{\distD}(\advA')$.
	
	Claim (ii) follows by noting that the view of $\advA$ is identical when run by $\sG^{\sprp}_0$ or by $\advA'$ in $\sG^{\dist}_{\distD,0}$.
	In $\sG^{\sprp}_0$, its construction queries are answered with the random permutation $F$.
	When it is run by $\advA'$ in $\sG^{\dist}_{\distD,0}$, these queries are answered with the tables $T$ and $T^{-1}$ which can be viewed as having just lazily sampled enough of a random permutation to respond to the given queries.
	In both cases, its primitive oracle is an independently chosen ideal cipher.
	
	Claim (i), follows by noting that the view of $\advA$ is identical when run by $\sG^{\sprp}_1$ or by $\advA'$ in $\sG^{\dist}_{\distD,1}$.
	In $\sG^{\sprp}_1$, construction queries are answered by the $\FX$ construction using the ideal cipher and keys $K_1$, $K_2$. 
	In the distribution, we first sample the responses to the construction queries and then construct the ideal cipher $f_1$ by picking $K_1$ and $K_2$ and setting $f_1$ to equal $f_0$ except for places where we reprogram it to be consistent with these construction queries.
	The construction will map $M_i$ to $Y_i$ for each $i$, which means that the condition $f_1(K_1,M_i\xor K_2)=Y_i\xor K_2$ should hold.
	These inputs and outputs are stored in the sets $\setI$ and $\setO$.
	The sets $\setI'$ and $\setO'$ store the inputs mapping to $\setO$ and outputs mapped to $\setI$ by $f_0(K_1,\cdot)$, respectively.
	Thus while making the above condition hold, we additionally reprogram $f_1$ so that elements of $\setI'\setminus\setI$ map to (random, non-repeating) elements of $\setO'\setminus\setO$.
	
	In particular, the uniformity of $f_0$ ensures that the map induced by $f_1(K_1,\cdot)$ between $\bits^n\setminus(\setI\cup\setI')$ and $\bits^n\setminus(\setO\cup\setO')$ is a random bijection.
	The last for loop samples a random bijection between $\setI'\setminus\setI$ and $\setO'\setminus\setO$.
	Because there are no biases in which values fall into these two cases among those of $\bits^n\setminus\setI$ and $\bits^n\setminus\setO$, this means the map between these two sets is a uniform bijection as desired.
	A more detailed probability analysis of Claim (i) is given 
	\ifnum\eprintversion=0
	\ifsubmission
		in \apref{app:na-proof-extra}.
	\else
		in the full version~\cite{fullversion}.
	\fi
	\else
		in \apref{app:na-proof-extra}.
	\fi
	
	Applying the bound on $\Adv^{\dist}_{\distD}(\advA')$ from \thref{thm:reg-ahu} gives us $$\Adv^{\sprp\kpa}_\FF(\advA)\leq 2p\sqrt{\Pr[\sG^{\guess}_\distD(\advA')]}$$ so we complete the proof by bounding this probability.
	Let $\mathsf{fw}$ denote the event that the $i$-th query of $\advA'$ in $\sH^{\distD}_0$ is to $f_0$ and let $(K,x)$ denote the measured value of this query so that
	\begin{displaymath}
		\Pr[\sH^{\guess}_\distD(\advA')]
		=
		\Pr[\mathsf{fw}]\cdot \Pr[(K,x)\in S \mid \mathsf{fw}]
		+
		\Pr[\neg\mathsf{fw}]\cdot \Pr[(K,x)\in S' \mid \neg\mathsf{fw}].
	\end{displaymath}

	A union bound over the different elements of $S$ gives
	\begin{align*}
		\Pr[(K,x)\in S \mid \mathsf{fw}]
		&\leq
		\sum_{j=1}^q
		\Pr[K=K_1 \land x \xor M_j = K_2 \mid \mathsf{fw}]\\
		&+
		\sum_{j=1}^q
		\Pr[K=K_1 \land G_0(K_1,x) \xor Y_j = K_2 \mid \mathsf{fw}].
	\end{align*}
	However, note that the view of $\advA$ in $\sH^{\distD}_0$ is independent of $K_1$ and $K_2$ so we get that
	\begin{displaymath}
		\Pr[(K,x)\in S \mid \mathsf{fw}]\leq 2q/2^{k+n}.
	\end{displaymath}
	Applying analogous analysis to the $\neg\mathsf{fw}$ case gives
	\begin{displaymath}
		\Pr[(K,x)\in S' \mid \neg\mathsf{fw}]\leq 2q/2^{k+n}
	\end{displaymath}
	and hence $\Pr[\sH^{\distD}_0(\advA')]\leq 2q/2^{k+n}$.
	Plugging this into our earlier inequality gives the stated bound.\qed
\end{proof}

\subsection{Adaptive Security of FFX}\label{sec:ffx-adaptive}
In this section we will prove the security of FFX (a variant of FX using a random oracle in place of the ideal cipher) against quantum adversaries making strictly classical queries to the construction.

\ifnum\eprintversion=0
	Formally, given a function family $\FF\in\Func(\kl{\FF},\il{\FF},\ol{\FF})$, we define the function family $\FFX[\FF]$ by $\FFX[\FF](K_1\concat K_2, x)= \FF_{K_1}(x \xor K_2)$.\footnote{The outer xor by $K_2$ used in $\FX$ is omitted because it is unnecessary for our analysis.}
	Here $|K_1|=\kl{\FF}$ and $|K_2|=\il{\FF}$ so $\kl{\FFX[\FF]}=\kl{\FF}+\il{\FF}$, $\il{\FFX[\FF]}=\il{\FF}$, and $\ol{\FFX[\FF]}=\ol{\FF}$.
	Let $k=\kl{\FF}$, $n=\il{\FF}$, and $m=\ol{\FF}$.
\else
	Formally, given a function family $\FF\in\Func(\kl{\FF},\il{\FF},\ol{\FF})$, we model the FFX construction by the function family $\FFX[\FF]$ by $\FFX[\FF](K_1\concat K_2, x)= \FF_{K_1}(x \xor K_2)$.\footnote{The outer xor by $K_2$ used in $\FX$ is omitted because it is unnecessary for our analysis.}
	Here $|K_1|=\kl{\FF}$ and $|K_2|=\il{\FF}$ so $\kl{\FFX[\FF]}=\kl{\FF}+\il{\FF}$, $\il{\FFX[\FF]}=\il{\FF}$, and $\ol{\FFX[\FF]}=\ol{\FF}$.
	Let $k=\kl{\FF}$, $n=\il{\FF}$, and $m=\ol{\FF}$.
\fi

\begin{theorem}\label{thm:ffx}
	Let $\advA$ be an order consistent quantum adversary which makes classical queries to $\procEv$ and consider $\FFX[\cdot]$ with the underlying function family modeled by a random oracle drawn from $\Func(k,n,m)$.
	Then 
	\begin{displaymath}
		\Adv^{\prf}_{\FFX}(\advA)\leq \sqrt{\frac{8(p+q)pq}{2^{k+n}}},
	\end{displaymath}
	where $p$ is the number of quantum oracle queries that $\advA$ makes to the random oracle and $q$ is the number of queries it makes to $\procEv$.
\end{theorem}
We can reasonably assume that $p>q$ so the dominant behavior of the above expression is $O\left(\sqrt{p^2q/2^{k+n}}\right)$.

The proof of this result proceeds via a sequence of hybrids which
gradually transition from the real world of $\sG^{\prf}_{\FFX}$ to the
ideal world.  Crucial to this sequence of hybrids are the technique of
Zhandry~\cite{C:Zhandry19} which, by viewing a space under dual bases,
allows one to simulate a random function using a sparse representation
table and to ``record'' the queries to the
function.
For the ideal
world, we can represent the random oracle and the random function
underlying $\procEv$ independently using such sparse representation
tables. With some careful modification, we are also able to represent
the real world's random oracle using a similar pair of sparse
representation tables as if it were two separate functions.  However,
in this case, the tables will be slightly non-independent in that if
the adversary queries $\procEv$ on an input $x$ and the random oracle
on $(K_1,x\xor K_2)$ then the results of the latter query is stored in
the $\procEv$ table, rather than the random oracle table.  Beyond this
minor consistency check (which we are only able to implement because
the queries to $\procEv$ are classical and so can be stored by
simulation), the corresponding games are identical.  Having done this
rewriting, we can carefully apply \thref{thm:reg-ahu} to bound the
ability of an adversary to distinguish between the two worlds by its
ability to trigger this consistency check.

As mentioned in \secref{sec:quant-back}, our application of \thref{thm:reg-ahu} here is somewhat atypical.
Our distribution over functions $\distD$ will be deterministic, but we are able to still extract a meaningful bound from this by taking advantage of our knowledge of the particular behavior of the adversary we apply \thref{thm:reg-ahu} with.

\begin{proof}
	In this proof we will consider a sequence of hybrid games $\sH_0$ through $\sH_9$.
	Of these games we will establish the following claims.
	\begin{enumerate}
		\item $\Pr[\sG^{\prf}_{\FFX,0}(\advA)]=\Pr[\sH_0]=\Pr[\sH_1]=\Pr[\sH_2]=\Pr[\sH_3]$
		\item $\Pr[\sG^{\prf}_{\FFX,1}(\advA)]=\Pr[\sH_9]=\Pr[\sH_8]=\Pr[\sH_7]=\Pr[\sH_6]=\Pr[\sH_5]=\Pr[\sH_4]$
		\item $\Pr[\sH_4]-\Pr[\sH_3]\leq	 \sqrt{8(p+q)pq/2^{\kl{\FF}+\il{\FF}}}$
	\end{enumerate}
	Combining these claims gives the desired result.
	
	In formally defining our hybrids we write the computation to be performed using the following quantum registers.
	\begin{itemize}
		\item $W$: The workspace of $\advA$. The adversary's final output is written into $W[1]$.
		\item $K$: The $k$-qubit register (representing the
                  function \emph{key/index}) $\advA$ uses when making
                  oracle queries to the random oracle.
		\item $X$: The $n$-qubit register (representing
                  function inputs) used when making oracle queries to the
                  random oracle or $\procEv$.
		\item $Y$: The $m$-qubit register (representing
                  function outputs) into which the results of oracle
                  queries are written.
		\item $H$: The $2^{k+n}\cdot m$-qubit register which stores the function defining the random oracle (initially via its truth table).
		\item $F$: The $2^n \cdot m$-qubit register which stores the function defining $\procEv$.
		\item $K_1$: The $k$-qubit register which stores the first key of the construction.
		\item $K_2$: The $n$-qubit register which stores the second key of the construction.
		\item $I$: The $\lceil\log p\rceil$-qubit register which tracks how many $\procEv$ queries $\advA$ has made.
		\item $\vec{X}=(\vec{X}_1,...,\vec{X}_p)$: The $p$ $n$-qubit registers used to store the classical queries that $\advA$ makes to $\procEv$.
	\end{itemize}
	
	We start by changing our perspective. 
	A quantum algorithm that makes classical queries to $\procEv$ can be modeled by thinking of a quantum algorithm that measures its $X$ register immediately before the query.
	(Because the behavior of $\procEv$ is completely classical at this point, we do not need to measure the $Y$ register as well.)
	Measuring the register $X$ is indistinguishable from using a CNOT operation to copy it into a separate register (i.e. xor-ing $X$ into the previously empty register $\vec{X}_I$ that will never again be modified).
	By incorporating this CNOT operation into the behavior of our hybrid game, we treat $\advA$ as an attacker that makes fully quantum queries to its oracles in the hybrid game.
	We think of $\advA$ as deferring all of measurements until the end of its computation.
	Because all that matters is its final output $W[1]$ we can have the game measure just that register and assume that $\advA$ does not internally make any measurements.
	The principle of deferred measurement ensures that the various changes discussed here do not change the behavior of $\advA$.
	This perspective change lets us use purely quantum analysis, rather than mixing quantum and classical.
	
	\heading{Claim 1.}  We start by considering the hybrids
        $\sH_0$ through $\sH_3$, defined in \figref{fig:ideal-hyb}
        which are all identical to the ideal world of
        $\sG^{\prf}_{\FFX}$. In these transitions we are applying the
        ideas of Zhandry~\cite{C:Zhandry19} to transition to
        representing the random functions stored in $H$ and $F$ by an
        all zeros table which is updated whenever the adversary makes
        a query.

        The hybrid $\sH_0$ is mostly just $\sG^{\prf}_{\FFX}$
        rewritten to use the registers indicated above.  
        So
        $\Pr[\sG^{\prf}_{\FFX,0}(\advA)]=\Pr[\sH_0]$ holds.
	
	\begin{figure}[t]
	\ifnum\eprintversion=0
        \renewcommand{\lwid}{0.5}
        \renewcommand{\mwid}{}
        \renewcommand{\rwid}{0.4}
	\else
        \renewcommand{\lwid}{0.45}
        \renewcommand{\mwid}{}
        \renewcommand{\rwid}{0.21}
	\fi
	\twoColsNoDivide{\lwid}{\rwid}
	{
	\underline{Games $\sH_0$, \hlg{$\sH_1$}}\\
	$\ro\getsr\Func(k,n,m)$\\
	$\textbf{F}\getsr\Func(0,n,m)$\\
	$\ket{H,F}\gets\ket{\ro,\mathbf{F}}$\\
	\hlg{$\ket{H,F}\gets \had\ket{0^{2^{k+n}\cdot m}, 0^{2^n\cdot m}}$}\\
	Run $\advA[{\procEv,\procRo}]$\\
	Measure $W[1]$, \hlg{$H$, $F$}\\
	Return $W[1]=1$\smallskip
	}
	{
	\underline{$\procEv(X,Y:I,\vec{X},F)$}\\
	$I\gets I+1 \mod 2^{|I|}$\\
	$\vec{X}_I \gets \vec{X}_I \xor X$\\
	$Y\gets F(X) \xor Y$\\
	Return $(X,Y : I,\vec{X},F)$\\[4pt]
	\underline{$\procRo(K,X,Y:H)$}\\
	$Y\gets H_K(X)\xor Y$\\
	Return $(K,X,Y : H)$\smallskip
	}
	\ifnum\eprintversion=0
        \renewcommand{\lwid}{0.6}
        \renewcommand{\mwid}{}
        \renewcommand{\rwid}{0.3}
	\else
        \renewcommand{\lwid}{0.45}
        \renewcommand{\mwid}{}
        \renewcommand{\rwid}{0.21}
	\fi
	\twoColsNoDivide{\lwid}{\rwid}
	{
	\underline{Games \fbox{$\sH_2$}, \hlg{$\sH_3$}}\\
	\fbox{$\ket{H,F}\gets \had\ket{0^{2^{k+n}\cdot m}, 0^{2^n\cdot m}}$}\\
	\fbox{Run $\advA[{\had^{Y,F}\circ \procFEv \circ \had^{Y,F}, \had^{Y,H}\circ \procFRo\circ \had^{Y,H}}]$}\\
	\hlg{$\ket{H,F}\gets \ket{0^{2^{k+n}\cdot m}, 0^{2^n\cdot m}}$}\\
	\hlg{Run $\advA[{\had^{Y}\circ \procFEv \circ \had^{Y},\had^{Y}\circ \procFRo\circ \had^{Y}}]$}\\
	\hlg{$\ket{H,F}\gets \had \ket{H,F}$}\\
	Measure $W[1]$, $H$, $F$\\
	Return $W[1]=1$\smallskip
	}
	{
	\underline{$\procFEv(X,Y:I,\vec{X},F)$}\\
	$I\gets I+1 \mod 2^{|I|}$\\
	$\vec{X}_I \gets \vec{X}_I \xor X$\\
	$F(X)\gets F(X)\xor Y$\\
	Return $(X,Y:I,\vec{X},F)$\\[4pt]
	\underline{$\procFRo(K,X,Y:H)$}\\
	$H_K(X)\gets H_K(X)\xor Y$\\
	Return $(K,X,Y:H)$\smallskip
	}
	\vspace{-2ex}
	\caption{Hybrid games $\sH_0$ through $\sH_3$ for the proof of \thref{thm:ffx} which are equivalent to the ideal world of $\sG^{\prf}_{\FFX}$. Highlighted or boxed code is only included in the correspondingly highlighted or boxed game.}
	\label{fig:ideal-hyb}
	\hrulefill
	\end{figure}
	
	Next consider $\sH_1$ which differs from $\sH_0$ only in the grey highlighted code which initializes $H$ and $F$ in the uniform superposition and then measures them at the end of execution. (Recall that the Hadamard transform applied to the all zeros state gives the uniform superposition.)
	Note that these register control, but are unaffected by the oracles
        $\procEv$ and $\procRo$.
	Because they are never modified while $\advA$ is
        executing, the principle of deferred measurement tells us that
        this modification is undetectable by
        $\advA$, giving $\Pr[\sH_0]=\Pr[\sH_1]$
	 
	Next consider $\sH_2$ which contains the boxed, but not the highlighted, code.
	This game uses the oracles $\procFEv$ and $\procFRo$, the Fourier versions of $\procEv$ and $\procRo$, which xor the $Y$ value of $\advA$'s query into the the register $F$ or $H$.
	Note that $\advA$'s access to these oracles is mitigated by $\had^{Y,F}$ on each query.
	The superscript here indicate that the Hadamard transform is being applied to the registers $Y$ and $F$.
	We have that ${\had^{Y,F}\circ \procFEv \circ \had^{Y,F}} = \procEv$ and $\had^{Y,H}\circ \procFRo\circ \had^{Y,H}=\procRo$ both hold.\footnote{
	This follows as a consequence of the following.
	Let $U_{\xor}$ and $U'_{\xor}$ be the unitaries which for $y,z\in\bits$ are defined by $U_{\xor}\ket{y,z}=\ket{y\xor z, z}$ and $U'_{\xor}\ket{y,z}=\ket{y,y\xor z}$.
	Then $\had \circ U_{\xor} \circ \had = U'_{\xor}$.}
	So $\Pr[\sH_1]=\Pr[\sH_2]$ because the adversary's oracles are identical.
	
	Next consider $\sH_3$ which contains the highlighted, but not
        the boxed, code.  For this transition, recall that
        $\had\circ\had$ is the identity operator.  So to transition to
        this game we can cancel the $\had$ operator used to initialize
        $H$ with the $\had^H$ operator applied before $\advA$'s first
        $\procFRo$ oracle query.  Similarly, we can cancel the
        $\had^H$ operation performed after any (non-final) $\procFRo$
        query with the $\had^H$ operation performed before the next
        $\procFRo$ query.  Finally, the $\had^H$ operation that would
        be performed after the final $\procFRo$ query is instead
        delayed to be performed immediately before $H$ is measured.
        (We could have omitted this operation and measurement entirely
        because all that matters at that point is the measurement of $W[1]$.)
        The $\had$ operators on $F$ are similarly changed.  Because
        $\advA$ does not have access to the $H$ and $F$ registers, we
        can indeed commute the $\had$ operators with $\advA$ in this manner
        without changing behavior.  Hence $\Pr[\sH_2]=\Pr[\sH_3]$, as
        desired.  Note that $H$ and $F$ independently store tables
        which are initialized to all zeros and then written into by
        the adversary's queries.
	
	\begin{figure}[t]
	\ifnum\eprintversion=0
        \renewcommand{\lwid}{0.5}
        \renewcommand{\mwid}{}
        \renewcommand{\rwid}{0.4}
	\else
        \renewcommand{\lwid}{0.36}
        \renewcommand{\mwid}{}
        \renewcommand{\rwid}{0.3}
	\fi
	\twoColsNoDivide{\lwid}{\rwid}
	{
	\underline{Games $\sH_9$, \hlg{$\sH_{8}$}}\\
	$\ro\getsr\Func(k,n,m)$\\
	$\mathbf{K_1}\getsr\bits^k$\\
	$\mathbf{K_2}\getsr\bits^n$\\
	$\ket{H,K_1,K_2}\gets\ket{\ro,\mathbf{K_1},\mathbf{K_2}}$\\
	\hlg{$\ket{H,K_1,K_2}\gets \had\ket{0^{2^{k+n}\cdot m}, 0^k, 0^n}$}\\
	Run $\advA[{\procEv,\procRo}]$\\
	Measure $W[1]$, \hlg{$H$, $K_1$, $K_2$}\\
	Return $W[1]=1$\smallskip
	}
	{
	\underline{$\procEv(X,Y:H,I,\vec{X},K_1,K_2)$}\\
	$I\gets I+1 \mod 2^{|I|}$\\
	$\vec{X}_I \gets \vec{X}_I \xor X$\\
	$Y\gets H_{K_1}(X \xor K_2)\xor Y$\\
	Return $(X,Y : H,I,\vec{X},K_1,K_2)$\\[4pt]
	\underline{$\procRo(K,X,Y : H)$}\\
	$Y\gets H_K(X)\xor Y$\\
	Return $(K,X,Y : H)$\smallskip
	}
	\twoColsNoDivide{\lwid}{\rwid}
	{
	\underline{Games \fbox{$\sH_{7}$}, \hlg{$\sH_{6}$}}\\
	\fbox{$\ket{H,K_1,K_2}\gets \had\ket{0^{2^{k+n}\cdot m}, 0^k, 0^n}$}\\
	\fbox{$O\gets {\had^{Y,H}\circ \procFEv \circ \had^{Y,H}}$}\\
	\fbox{$O'\gets \had^{Y,H}\circ \procFRo\circ \had^{Y,H}$}\\
	\fbox{Run $\advA[O,O']$}\\
	\hlg{$\ket{K_1,K_2}\gets \had \ket{0^k, 0^n}$}\\
	\hlg{$\ket{H}\gets \ket{0^{2^{k+n}\cdot m}}$}\\
	\hlg{Run $\advA[{\had^{Y}\circ \procFEv \circ \had^{Y},\had^{Y}\circ \procFRo\circ \had^{Y}}]$}\\
	\hlg{$\ket{H}\gets \had \ket{H}$}\\
	Measure $W[1]$, $H$, $K_1$, $K_2$\\
	Return $W[1]=1$\smallskip
	}
	{
	\underline{$\procFEv(X,Y:H,I,\vec{X},K_1,K_2)$}\\
	$I\gets I+1 \mod 2^{|I|}$\\
	$\vec{X}_I \gets \vec{X}_I \xor X$\\
	$H_{K_1}(X \xor K_2)\gets H_{K_1}(X \xor K_2)\xor Y$\\
	Return $(X,Y : H,I,\vec{X},K_1,K_2)$\\[4pt]
	\underline{$\procFRo(K,X,Y : H)$}\\
	$H_K(X)\gets H_K(X)\xor Y$\\
	Return $(K,X,Y : H)$
	}
	\vspace{-2ex}
	\caption{Hybrid games $\sH_9$ through $\sH_6$ for the proof of \thref{thm:ffx} which are equivalent to the real world of $\sG^{\prf}_{\FFX}$. Highlighted or boxed code is only included in the correspondingly highlighted or boxed game.}
	\label{fig:real-hyb}
	\hrulefill
	\end{figure}

	\heading{Claim 2.}  We now consider the hybrids $\sH_4$
        through $\sH_9$ (starting from $\sH_9$), which are defined in
        \figref{fig:real-hyb} and \figref{fig:real-hyb2} using similar
        ideas as in the transition from $\sH_0$ through $\sH_3$.  As
        we will justify, these games are all equivalent to the real
        world of $\sG^{\prf}_{\FFX}$.
	
	First $\sH_9$ rewrote the real world of $\sG^{\prf}_{\FFX}$ to
        use our specified registers and to record queries into
        $\vec{X}$ in $\procEv$.  Then in $\sH_8$, rather than
        sampling $H$, $K_1$, and $K_2$ uniformly at the beginning of
        the game we put them in the uniform superposition and measure
        them at the end of the game.  For $\sH_7$ we replace our
        oracles that xor into the adversary's $Y$ register with
        oracles that xor into the $H$ register using Hadamard
        operations, some of which we then cancel out to transition to
        $\sH_6$.  The same arguments from Claim 1 of why these sorts
        of modifications do not change the behavior of the game apply
        here and so
        $\Pr[\sG^{\prf}_{\FFX,1}(\advA)]=\Pr[\sH_9]=\Pr[\sH_8]=\Pr[\sH_7]=\Pr[\sH_6]$.
	
	\begin{figure}[t]
	\ifnum\eprintversion=0
        \renewcommand{\lwid}{0.54}
        \renewcommand{\mwid}{}
        \renewcommand{\rwid}{0.39}
	\else
        \renewcommand{\lwid}{0.43}
        \renewcommand{\mwid}{}
        \renewcommand{\rwid}{0.3}
	\fi
	\twoColsNoDivide{\lwid}{\rwid}
	{
	\underline{Games \fbox{$\sH_{5}$}, \hlg{$\sH_{4}$}}\\
	{$\ket{K_1,K_2}\gets \had \ket{0^k, 0^n}$}\\
	{$\ket{H,F}\gets \ket{0^{2^{k+n}\cdot m}, 0^{2^n\cdot m}}$}\\
	\fbox{$O\gets  \had^{Y}\circ \dec\circ \procFEv \circ \dec \circ \had^{Y}$}\\
	\fbox{$O'\gets  \had^{Y}\circ \dec \circ \procFRo \circ  \dec \circ \had^{Y}$}\\
	\fbox{Run $\advA[{O,O'}]$}\\
	\hlg{Run $\advA[{ \had^{Y}\circ \procFEv' \circ \had^{Y}, \had^{Y}\circ \procFRo'\circ \had^{Y}}]$}\\
	$\ket{H,F,I,\vec{X},K_1,K_2}\gets \dec\ket{H,F,I,\vec{X},K_1,K_2}$\\
	{$\ket{H}\gets \had \ket{H}$}\\
	Measure $W[1]$, $H$, $K_1$, $K_2$\\
	Return $W[1]=1$\\[4pt]
	\underline{$\procFRo'(K,X,Y : H,F,I,\vec{X},K_1,K_2)$}\\
	If $K=K_1$ and $X\xor K_2 \in \sett{\vec{X}_1,\dots,\vec{X}_I}$ then\\
	\ind \commentt{Input is bad}\\
	\ind $F(X\xor K_2)\gets F(X\xor K_2) \xor Y$\\
	Else\\
	\ind $H_K(X)\gets H_K(X)\xor Y$\\
	Return $(K,X,Y : H,F,I,\vec{X},K_1,K_2)$\smallskip
	}
	{
	\underline{$\procFEv(X,Y:H,I,\vec{X},K_1,K_2)$}\\
	$I\gets I+1 \mod 2^{|I|}$\\
	$\vec{X}_I \gets \vec{X}_I \xor X$\\
	$H_{K_1}(X \xor K_2)\gets H_{K_1}(X \xor K_2)\xor Y$\\
	Return $(X,Y : H,I,\vec{X},K_1,K_2)$\\[4pt]
	\underline{$\procFRo(K,X,Y : H)$}\\
	$H_K(X)\gets H_K(X)\xor Y$\\
	Return $(K,X,Y : H)$\\[23pt]
	\underline{$\procFEv'(X,Y:H,F,I,\vec{X},K_1,K_2)$}\\
	$I\gets I+1 \mod 2^{|I|}$\\
	$\vec{X}_I \gets \vec{X}_I \xor X$\\
	$\mathsf{bool}_1\gets (X\not\in\{\vec{X}_1,\dots,\vec{X}_{I-1}\})$\\
	$\mathsf{bool}_2\gets (H_{K_1}(X \xor K_2)\neq 0^m)$\\ 
	$\mathsf{bool}_3\gets (F(X)\neq 0^m)$\\
	If $\mathsf{bool}_1$ and ($\mathsf{bool}_2$ or $\mathsf{bool}_3$) then\\
	\ind \commentt{Input is bad}\\
	\ind $F'(X)\gets F(X)$\\
	\ind $F(X)\gets H_{K_1}(X \xor K_2)$\\
	\ind $H_{K_1}(X \xor K_2) \gets F'(X)$\\
	$F(X) \gets F(X)\xor Y$\\
	Return $(X,Y : H,I,\vec{X},K_1,K_2)$\smallskip
	}
	\vspace{-2ex}
	\caption{Hybrid games $\sH_5$ and $\sH_4$ for the proof of \thref{thm:ffx} which are equivalent to the real world of $\sG^{\prf}_{\FFX}$. Highlighted or boxed code is only included in the correspondingly highlighted or boxed game. Unitary $\dec$ is define  in the text.}
	\label{fig:real-hyb2}
	\hrulefill
	\end{figure}

	Our next transitions are designed to make the current game's oracles identical with those of $\sH_3$, except on some ``bad'' inputs.
	In $\sH_6$ we have a single all zeros table $H$ which gets written into by queries that $\advA$ makes to either of its oracles, while in $\sH_3$ the oracles separately wrote into either $H$ or $F$. 
	For $\sH_5$ we will similarly separate the single table $H$ into separate tables $H$ and $F$.
	However, we cannot keep them completely independent, because if the adversary queries $\procFEv$ with $X=x$ and $\procFRo$ with $(K,X)=(K_1, x\xor K_2)$ then both of these operations would be writing into the same table location in $\sH_6$.
	Consider the following unitary $\dec$ which acts on registers $H$ and $F$ and is controlled by the registers $I$, $\vec{X}$, $K_1$, and $K_2$.
	We will think of this unitary as transitioning us between a representation of $H$ as a single table (with an all-zero $F$ table) and a representation of it divided between $H$ and $F$.
	\oneColNoBox{
	\underline{$\dec(H,F,I,\vec{X},K_1,K_2)$}\\
	For $x\in \sett{\vec{X}_1,\dots,\vec{X}_I}$ do\\
	\ind $F'(x)\gets F(x)$\\
	\ind $F(x)\gets H_{K_1}(x \xor K_2)$\\
	\ind $H_{K_1}(x \xor K_2) \gets F'(x)$\\
	Return $(H,F,I,\vec{X},K_1,K_2)$\smallskip
	}
	In words, $\dec$ swaps $F(x)$ and $H_{K_1}(x\xor K_2)$ for each $x$ that has been previous queried to $\procFEv$ (as stored by $\vec{X}$ and $I$).
	Note that $\dec$ is its own inverse.
	In $\sH_5$ we (i) initialize the table $F$ as all zeros, (ii) perform $\dec$ before and after each oracle query, and (iii) perform $\dec$ after $\advA$ has executed.
	We verify that $H$ has the same value in $\sH_5$ that it would have had in $\sH_6$ during each oracle query and at the end before measurement.
	The application of $\dec$ before the first oracle query does nothing (because $I=0$) so $H$ is all zeros for this query as required.
	As we've seen previously with $\had$, we can commute $\dec$ with the operations of $\advA$ because $\dec$ only acts on registers outside of the adversary's control.
	We can similarly commute $\dec$ with $\had^Y$.
	Hence the $\dec$ operation after every non-final oracle query can be seen to cancel with the $\dec$ operation before the following oracle query.
	The $\dec$ operation after the final oracle query cancels with the $\dec$ operation performed after $\advA$ halts execution.
	Hence, $\Pr[\sH_6]=\Pr[\sH_5]$ as claimed.
	
	For the transition to $\sH_4$ let us dig into how our
        two-table representation in $\sH_5$ works so that we can
        incorporate the behavior of $\dec$ directly into the oracles. For simplicity of
        notation in discussion, let $\tH(x)$ denote $H_{K_1}(x \xor K_2)$.  First
        note that in between oracle queries the two tables
        representation of $\sH_5$ will satisfy the property that for each
        $x\in \sett{\vec{X}_1,\dots,\vec{X}_I}$ we will have
        $\tH(x)=0^{\ol{\FF}}$ and for all other $x$ we will have that
        $F(x)=0^{\ol{\FF}}$.\footnote{More precisely, the
          corresponding registers hold superpositions over tables
          satisfying the properties we discuss.}  This is the case
        because after each query we have applied $\dec$ to an $H$ which
        contains the same values it would have in $\sH_6$ and an $F$
        which is all zeros.
	
	Now consider when a $\dec\circ\procFEv\circ\dec$ query is executed with some $X$ and $Y$.
	If $X\in \sett{\vec{X}_1,\dots,\vec{X}_I}$, then $F(X)$ and $\tH(X)$ are swapped, $Y$ is xored into $\tH(X)$, then finally $F(X)$ and $\tH(X)$ are swapped back.
	Equivalently, we could have xored $Y$ directly into $F(X)$ and skipped the swapping around.
	If $X\not\in \sett{\vec{X}_1,\dots,\vec{X}_I}$, then $Y$ is xored into $\tH(X)$ before $F(X)$ and $\tH(X)$ are swapped.
	If $\tH(X)=0^{\ol{F}}$ beforehand, then we could equivalently have xored $Y$ directly into $F(X)$ and skipped the swapping around (because $F(X)=0^{\ol{\FF}}$ must have held from our assumption on $X$).
	If $\tH(X)\neq0^{\ol{F}}$ beforehand, then we could equivalently could have swapped $\tH(X)$ and $F(X)$ first, then xored $Y$ into $F(X)$.
	The equivalent behavior we described is exactly the behavior captured by the oracle $\procFEv'$ which is used in $\sH_4$ in place of $\dec\circ\procFEv\circ\dec$.
	It checks if $X\not\in \sett{\vec{X}_1,\dots,\vec{X}_I}$ ($\mathsf{bool}_1$) and $\tH(X)\neq 0^m$ ($\mathsf{bool_2}$), performing a swap if so.
	Then $Y$ is xored into $F(X)$.
	A swap is also performed if $X\not\in \sett{\vec{X}_1,\dots,\vec{X}_{I-1}}$ and $F(X)\neq 0^m$, however this case is impossible from our earlier observation that $F(x)=0^{m}$ when $X$ is not in $\vec{X}$.
	This second case was added only to ensure that $\procFEv'$ is a permutation.
	
	Similarly, consider when a $\dec\circ\procFRo\circ \dec$ query is executed with some $K$, $X$, $Y$.
	Any swapping done by $\dec$ uses $H_{K_1}$, so when $K\neq K_1$ this just xors $Y$ into $H_K(X)$.
	If $X\xor K_2$ is not in $\sett{\vec{X}_1,\dots,\vec{X}_I}$, then $H_K(X)$ is unaffected by the swapping so again this just xors $Y$ into $H_K(X)$.
	When $K= K_1$ and $X\xor K_2\in\sett{\vec{X}_1,\dots,\vec{X}_I}$, first $H_K(X)$ would have been swapped with $F(X\xor K_2)$, then $Y$ would be xored into $H_K(X)$, then $H_K(X)$ and $F(X\xor K_2)$ would be swapped back.
	Equivalently, we could just have xored $Y$ into $F(X\xor K_2)$ and skipped the swapping.
	This behavior we described is exactly the behavior captured by the oracle $\procFRo'$ which is used in $\sH_4$ in place of $\dec\circ\procFRo\circ\dec$.	
	
	We have just described that on the inputs we care about $\procFEv'$ behaves identically to $\dec\circ\procFEv\circ\dec$ and $\procFRo'$ behaves identically to $\dec\circ\procFRo\circ\dec$. Hence $\Pr[\sH_5]=\Pr[\sH_4]$, completing this claim.
	
	\begin{figure}[t]
	\ifnum\eprintversion=0
        \renewcommand{\lwid}{0.55}
        \renewcommand{\mwid}{}
        \renewcommand{\rwid}{0.37}
	\else
        \renewcommand{\lwid}{0.41}
        \renewcommand{\mwid}{}
        \renewcommand{\rwid}{0.28}
	\fi
	\twoColsNoDivide{\lwid}{\rwid}
	{
	\underline{Games \fbox{$\widetilde{\sH}_3$}, \hlg{$\widetilde{\sH}_{4}$}}\\
	{$\ket{K_1,K_2}\gets \had \ket{0^k, 0^n}$}\\
	{$\ket{H,F}\gets \ket{0^{2^{k+n}\cdot m}, 0^{2^n\cdot m}}$}\\
	{Run $\advA[{ \had^{Y}\circ \procFEv \circ \had^{Y}, \had^{Y}\circ \procFRo\circ \had^{Y}}]$}\\
	Measure $W[1]$\\
	Return $W[1]=1$\\[4pt]
	\underline{$\procFRo(K,X,Y : H,F,I,\vec{X},K_1,K_2)$}\\
	If $K=K_1$ and $X\xor K_2 \in \sett{\vec{X}_1,\dots,\vec{X}_I}$ then\\
	\ind \commentt{Input is bad}\\
	\ind \hlg{$F(X\xor K_2)\gets F(X\xor K_2) \xor Y$}\\
	\ind \fbox{$H_K(X)\gets H_K(X)\xor Y$}\\
	Else\\
	\ind $H_K(X)\gets H_K(X)\xor Y$\\
	Return $(K,X,Y : H,F,I,\vec{X},K_1,K_2)$\smallskip
	}
	{
	\underline{$\procFEv(X,Y:H,F,I,\vec{X},K_1,K_2)$}\\
	$I\gets I+1 \mod 2^{|I|}$\\
	$\vec{X}_I \gets \vec{X}_I \xor X$\\
	$\mathsf{bool}_1\gets (X\not\in\{\vec{X}_1,\dots,\vec{X}_{I-1}\})$\\
	$\mathsf{bool}_2\gets (H_{K_1}(X \xor K_2)\neq 0^m)$\\ 
	$\mathsf{bool}_3\gets (F(X)\neq 0^m)$\\
	If $\mathsf{bool}_1$ and ($\mathsf{bool}_2$ or $\mathsf{bool}_3$)\\
	\ind \commentt{Input is bad}\\
	\ind \hlg{$F'(X)\gets F(X)$}\\
	\ind \hlg{$F(X)\gets H_{K_1}(X \xor K_2)$}\\
	\ind \hlg{$H_{K_1}(X \xor K_2) \gets F'(X)$}\\
	$F(X) \gets F(X)\xor Y$\\
	Return $(X,Y : H,I,\vec{X},K_1,K_2)$\smallskip
	}
	\vspace{-2ex}
	\caption{Hybrid games $\widetilde{\sH}_3$ and $\widetilde{\sH}_4$ for the proof of \thref{thm:ffx} which are rewritten versions of $\sH_3$ and $\sH_4$ to emphasize that their oracles are identical-until-bad. Highlighted or boxed code is only included in the correspondingly highlighted or boxed game.}
	\label{fig:show-iub}
	\hrulefill
	\end{figure}
	
	\heading{Claim 3.}
	To compare hybrids $\sH_3$ and $\sH_4$ we will note their oracles only differ on a small number of inputs (in particular those labelled as bad by comments in our code) and then apply \thref{thm:reg-ahu} to bound the difference between them.
	To aid in this we have rewritten them as $\widetilde{\sH}_3$ and $\widetilde{\sH}_4$ in \figref{fig:show-iub}.
	For both we have removed some operations performed after $\advA$ halted its execution for cleanliness because these  operations on registers other than $W[1]$ cannot affect the probability that it is measured to equal $1$.
	So we have $\Pr[\widetilde{\sH}_3]=\Pr[\sH_3]$ and $\Pr[\widetilde{\sH}_4]=\Pr[\sH_4]$.
	
	Let $\procFEv_3$ and $\procFRo_3$ be the permutations defining the corresponding oracle in $\widetilde{\sH}_3$.
	Define $\procFEv_4$ and $\procFRo_4$ analogously.
	These permutations differ only on the inputs we referred to as bad. 
	So let $S$ denote the set of bad $X,H,F,I,\vec{X},K_1,K_2$ for $\procFEv$ (i.e. those for which $\mathsf{bool}_1$ and either $\mathsf{bool}_2$ or $\mathsf{bool}_3$ hold).
	Let $S'$ denote the set of bad $K,X,H,F,I,\vec{X},K_1,K_2$ for $\procFRo$  (i.e. those for which $K=K_1$ and $X\xor K\in\sett{\vec{X}_1,\dots,\vec{X}_I}$).
	Let $\distD$ denote the distribution which always outputs $(S,S',\procFEv_3,\procFRo_3,\procFEv_4,\procFRo_4,\varepsilon)$.
	Clearly this is a valid distribution for \thref{thm:reg-ahu} by our choice of $S$ and $S'$.
	
	Now we can define an adversary $\advA'$ for $\sG^{\dist}_{\distD,b}$ which simulates the view of $\advA$ in $\widetilde{\sH}_{3+b}$ by locally running the code of that hybrid except for during oracle queries when it uses its $f_b$ oracle to simulate $\procFEv$ and $f'_b$ oracle to simulate $\procFRo$.
	Because the simulation of these views are perfect we have that
	\begin{align*}
		\Pr[\widetilde{\sH}_3]-\Pr[\widetilde{\sH}_4]
		=\Adv^{\dist}_{\distD}(\advA')
		\leq 2(p+q)\sqrt{\Adv^{\guess}_{\distD}(\advA')}
	\end{align*}
 	where the inequality follows from \thref{thm:reg-ahu}, noting that $\advA'$ makes $p+q$ oracle queries.
	
	To complete the proof we bound $\Adv^{\guess}_{\distD}(\advA')$.
	In the following probability calculation we use $x$ and $i$ to denote random variables taking on the values the corresponding variables have at the end of an execution of $\sG^{\guess}_{\distD}(\advA)$.
	Let $\mathcal{S}$ denote a random variable which equals $S$ if the measured query is to $f_0$ and equals $S'$ otherwise.
	Then conditioning over each possible value of $i$ gives
	\begin{align*}
		\Adv^{\guess}_{\distD}(\advA')
		=	\Pr[x \in \mathcal{S}]
		&= \sum_{j=1}^{p+q} \Pr[x \in \mathcal{S} \mid i=j] \Pr[i=j]\\
		&= (p+q)^{-1}\sum_{j=1}^{p+q} \Pr[x \in \mathcal{S} \mid i=j].
	\end{align*}
	
	Because $\advA$ is order consistent we can pick disjoint sets $E$ and $R$ with $E\cup R = \sett{1,\dots,p+q}$ such that $i\in E$ means the $i$-th query is to $\advA$'s $\procFEv$ oracle and $i\in R$ means that the $i$-th query is to its $\procFRo$ oracle.
	Note that $|E|=q$ and $|R|=p$.
	The view of $\advA$ (when run by $\advA'$) in
        $\sG^{\guess}_{\distD}$ matches its view in
        $\widetilde{\sH}_3$ so, in particular, it is independent of
        $K_1$ and $K_2$. 
	Hence we can think of these keys being chosen at random at the end of execution when analyzing the probability of $x\in \mathcal{S}$.
	
	For a $\procFRo$ query, the check for bad inputs is if $K_1=K$ and $K_2\in\sett{X\xor \vec{X}_1,\dots,X\xor \vec{X}_I}$.
	The variable $I$ is counting the number of $\procFEv$ queries made so far, so $I<q$.
	By a union bound,
	\begin{displaymath}
		\Pr[x\in S' \mid i=j] \leq q/2^{\kl{\FF}+\il{\FF}}
	\end{displaymath}
	when $j\in R$.
	
	For a $\procFEv$ query, the check for bad inputs is if $X\not\in\{\vec{X}_1,\dots,\vec{X}_{I-1}\}$ and $H_{K_1}(X\xor K_2)$ is non-zero.\footnote{Note that the other bad inputs, for which $X\not\in\{\vec{X}_1,\dots,\vec{X}_{I-1}\}$ and $F(X)$ is non-zero, will never occur.}
	In $\widetilde{\sH}_3$, each query to $\procFRo$ can make a single entry of $H$ non-zero so it will never have more than $p$ non-zero entries.
	By a union bound, 
	\begin{displaymath}
		\Pr[x\in S \mid i=j] \leq p/2^{\kl{\FF}+\il{\FF}}
	\end{displaymath}
	when $j\in E$.
	
	The proof is then completed by noting
	\begin{align*}
	\sum_{j=1}^{p+q} \Pr[x\in \mathcal{S} \mid i=j]
	&=\sum_{j\in R} \Pr[x\in S' \mid i=j] + \sum_{j\in E} \Pr[x\in S \mid i=j]\\
	&\leq p(q/2^{\kl{\FF}+\il{\FF}}) + q(p/2^{\kl{\FF}+\il{\FF}})
	=2pq/2^{\kl{\FF}+\il{\FF}}
	\end{align*}
	and plugging in to our earlier expression.\qed
\end{proof}


\section{Double Encryption}\label{sec:double-enc}
In this section we prove the security of the double encryption key-length extension technique against fully quantum attacks.
Our proof first reduces this to the ability of a quantum algorithm to solve the list disjointness problem and then extends known query lower bounds for element distinctness to list disjointness (with some modifications).

The double encryption blockcipher is constructed via two sequential application of an underlying blockcipher.
Formally, given a blockcipher $\EF\in\Icm(\kl{\EF},\bl{\EF})$, we define the double encryption blockcipher $\DE[\EF]$ by $\DE[\EF](K_1\concat K_2, x)= \EF_{K_2}(\EF_{K_1}(x))$.
Here $|K_1|=|K_2|=\kl{\EF}$ so $\kl{\DE[\EF]}=2\kl{\EF}$ and $\bl{\DE[\EF]}=\bl{\EF}$.
Its inverse can be computed as $\DE[\EF]^{-1}(K_1\concat K_2, x)= \EF^{-1}_{K_1}(\EF^{-1}_{K_2}(x))$.

Classically, the meet-in-the-middle attack~\cite{diffie1977exhaustive,merkle1981security} shows that this construction achieves essentially the same security a single encryption.
In the quantum setting, this construction was recently considered by Kaplan~\cite{kaplan2014quantum}.
They gave an attack and matching security bound for the key-recovery security of double encryption.
This leaves the question of whether their security result can be extended to cover full SPRP security, which we resolve by the main theorem of this section.
This theorem is proven via a reduction technique of Tessaro and Thiruvengadam~\cite{TCC:TesThi18} which they used to establish a (classical) time-memory tradeoff for the security of double encryption, by reducing its security to the list disjointness problem and conjecturing a time-memory tradeoff for that problem.

\heading{Problems and languages.}
In addition to the list disjointness problem ($\ld$), we will also consider two versions of the element distinctness problem ($\ed$, $\oed$).
In general, a problem $\pp$ specifies a relation $\rel$ on set on instances $\bigi$ (i.e. $\rel$ is function which maps an instance $L\in\bigi$ and witness $w$ to a decision $\rel(L,w)\in\bits$).
This relation induces a language $\lang=\{L\in\bigi : \exists w, \rel(L,w)=1\}$.
Rather than think of instances as bit strings, we will think of them as functions (to which decision and search algorithms are given oracle access).
To restrict attention to functions of specific sizes we let $\lang({\dom,\rng})=\lang\cap\bigi({\dom,\rng})$ where $\bigi({\dom,\rng})$ denotes the restriction of $\bigi$ to functions $L:[\dom]\to[\rng]$, where $D\leq R$.
To discuss instances not in the language we let $\lang'=\bigi\setminus\lang$ and ${\lang'}({\dom,\rng})=\bigi({\dom,\rng})\setminus\lang$.

Problems have decision and search versions.
The goal of a decision algorithm is to output 1 (representing ``acceptance'') on instances in the language and 0 (representing ``rejection'') otherwise.
Relevant quantities are the minimum probability $P^1$ of accepting an instance in the language, the maximum probability $P^0$ of accepting an instance not in the language, and the error rater $E$ which are formally defined by
\begin{align*}
	P_{\dom,\rng}^1(\advA) &= \min_{L\in \lang({\dom,\rng})}\Pr[\advA[L]=1],\ 
	P_{\dom,\rng}^0(\advA) = \max_{L\in {\lang}'({\dom,\rng})}\Pr[\advA[L]=1]\\
	E_{\dom,\rng}(\advA) &= \max \{1- P_{\dom,\rng}^1(\advA), P_{\dom,\rng}^0(\advA)\}.
\end{align*}
We define the decision $\pp$ advantage of $\advA$ by $\Adv^{\pp}_{\dom,\rng}(\advA)=P_{\dom,\rng}^1(\advA)-P_{\dom,\rng}^0(\advA)$.
In non-cryptographic contexts, instead of looking at the difference in probability that inputs that are in or out of the language are accepted, one often looks at how far these are each from $1/2$.
This motivates the definition $\Adv^{\pp\dc}_{\dom,\rng}(\advA)=\min \{2P_{\dom,\rng}^1(\advA)-1, 1-2P_{\dom,\rng}^0(\advA)\}=1-2E_{\dom,\rng}(\advA)$.

The goal of a search algorithm is to output a witness for the instance.
We define its advantage to be the minimum probability it succeeds, i.e., $\Adv^{\pp\srch}_{\dom,\rng}(\advA)=\min_{L\in\lang(D,R)}\Pr[\rel(L,\advA[L])=1]$.

\begin{figure*}[t]
        \begin{center}
        {%
        \begin{tabular}{|c|c|c|}
        \hline
        	Problem & Witness & Promise\\
        \hline\hline
    		$\ed$ & $x\neq y \suchthat L(x)=L(y)$ & - \\[2pt]
    		$\oed$ & $x\neq y \suchthat L(x)=L(y)$ & At most one witness. \\[2pt]
    		$\ld$ & $x,y\suchthat L_0(x)=L_1(y)$ & At most one witness. Injective $L_0,L_1$.\\[2pt]
    	\hline
        \end{tabular}}	
        \end{center}
\vspace{-1.5em}
\caption{Summary of the element distinctness and list disjointness problems we consider.}
\label{fig:languagese}
\hrulefill
\end{figure*}

\heading{Example problems.}
The \emph{list disjointness} problem asks how well an algorithm can distinguish between the case that is give (oracle access to) two lists which are disjoint or have one element in common.
In particular, we interpret an instance $L$ as the two functions $L_0,L_1:[\lfloor\dom/2\rfloor]\to[\rng]$ defined by $L_b(x)=L(x + b\lfloor \dom/2\rfloor)$.
Let $\mathcal{S}_n$ denote the set of $L$ for which $L_0$ and $L_1$ are injective and which have $n$ elements in common, i.e. for which $|\sett{L_0(1),\dots,L_0(\lfloor\dom/2\rfloor)} \cap \sett{L_1(1),\dots,L_1(\lfloor\dom/2\rfloor)} | = n$.
Then $\ld$ is defined by the relation $\rel$ which on input $(L,(x,y))$ returns 1 iff $L_0(x)=L_1(y)$ and the instance set $\bigi = \mathcal{S}_0 \cup \mathcal{S}_1$ (i.e., the promise that there is at most one element in common and that the lists are individually injective).
The search version of list disjointness is sometimes referred to as claw-finding.

The \emph{element distinctness} problem asks how well an algorithm can detect whether all the elements in a list are distinct.
Let $\mathcal{S}'_n$ denote the set of $L$ which have $n$ collision pairs, i.e. for which $|\sett{\sett{x,y} : x\neq y, L(x)=L(y)} | = n$.
Then $\ed$ is defined by the relation $\rel$ which on input $(L,(x,y))$ returns 1 iff $x\neq y$ and $L(x)=L(y)$ with the instance set $\bigi=\bigcup_{n=0}^\infty S'_n$ consisting of all functions.
 We let $\oed$ denote restricting $\ed$ to $\bigi=\mathcal{S}'_0 \cup \mathcal{S}'_1$ (i.e., the promise that there is at most one repetition in the list).

\subsection{Security result}
The following theorem shows that an attacker achieving constant
advantage must make $\Omega(2^{2k/3})$ oracle queries (ignoring log
terms).  Our bound is not tight for a lower parameter regimes, though
future work may establish better bounds for list disjointness in these
regimes.

\begin{theorem}\label{thm:de-security}
	Consider $\DE[\cdot]$ with the underlying blockcipher modeled by an ideal cipher drawn from $\Icm(k,n)$.
	Let $\advA$ be a quantum adversary which makes at most $q$ queries to the ideal cipher.
	Then 
	\begin{displaymath}
		\Adv^{\sprp}_{\DE}(\advA)\leq 11\sqrt[6]{(q\cdot k\lg k)^3/2^{2k}} + 1/2^k.
	\end{displaymath} 
\end{theorem}

As mentioned earlier, our proof works by first reducing the security of double encryption against quantum queries to the security of the list disjointness problem against quantum queries.
This is captured in the following theorem which we prove now.

\begin{theorem}\label{thm:de-reduction}
	Consider $\DE[\cdot]$ with the underlying blockcipher modeled by an ideal cipher drawn from $\Icm(k,n)$.
	Let $\advA$ be a quantum adversary which makes at most $q$ queries to the ideal cipher.
	Then for any $\rng\geq 2^k$ we can construct $\advA'$ making at most $q$ oracle queries such that
	\begin{displaymath}
		\Adv^{\sprp}_{\DE}(\advA)\leq\Adv^{\ld\dc}_{2^k,\rng}(\advA') + 1/2^k.
	\end{displaymath} 
\end{theorem}

We state and prove a bound on $\Adv^{\ld\dc}$ in \secref{sec:hard-ld}.
Our proof applies the same reduction technique as Tessaro and Thiruvengadam~\cite{TCC:TesThi18}, we are verifying that it works quantumly as well.

\begin{proof}
For $b\in\bits$, let $\sH_{b}$ be defined to be identical to $\sG^{\sprp}_{\DE,b}$ except that $K_2$ is chosen uniformly from $\bits^{k}\setminus \sett{K_1}$ rather than from $\bits^{k}$.
This has no effect when $b=0$ because the keys are not used, so $\mathsf{Pr}[\sG^{\sprp}_{\DE,0}(\advA)]=\Pr[\sH_0]$.
When $b=1$ there was only a $1/2^k$ chance that $K_2$ would have equalled $K_1$ so $\mathsf{Pr}[\sG^{\sprp}_{\DE,1}(\advA)]\leq\Pr[\sH_1]+1/2^k$.

Now we define a decision algorithm $\advA'$ for $\ld$ which uses its input lists to simulate a view for $\advA$.
When the lists are disjoint $\advA$'s view will perfectly match that of $\sH_0$ and when the lists have exactly one element in common $\advA$'s view will perfectly match that of $\sH_1$.
Hence we have $\Pr[\sH_1] = \min_{(L,L')\in\mathcal{L}^{\kappa,k'}_1} \Pr[\sG^{\ldis}_{L,L'}(\advA)]$ and $\Pr[\sH_0] = \max_{(L,L')\in\mathcal{L}^{\kappa,k'}_0} \Pr[\sG^{\ldis}_{L,L'}(\advA)]$, so $\Pr[\sH_1] -\Pr[\sH_0] =  \Adv^{\ldis}_{2^k,k'}(\advA')$ which gives the claimed bound.

\begin{figure}[t]
\ifnum\eprintversion=0
        \renewcommand{\lwid}{0.28}
        \renewcommand{\mwid}{0.31}
        \renewcommand{\rwid}{0.32}
	\else
        \renewcommand{\lwid}{0.21}
        \renewcommand{\mwid}{0.23}
        \renewcommand{\rwid}{0.24}
	\fi
\threeColsNoDivide{\lwid}{\mwid}{\rwid}
{
\underline{Adversary $\advA'[L]$}\\
$\rho\getsr\Icm(0,k)$\\
$\pi\getsr\Icm(0,n)$\\
$F\getsr\Icm(\lceil\lg R\rceil,n)$\\
$b'\getsr\advA[{\pi,\pi^{-1},\procIc,\procIcInv}]$\\
Return $b'$\smallskip
}
{
\underline{$\procIc(K,X,Y : \rho, \pi, F)$}\\
$(i,j)\gets\rho(K)$\\
If $i=0$ then\\
\ind $Y \gets Y \xor F_{L_0(j)}(X)$\\
Else\\
\ind $Y \gets Y \xor \pi(F^{-1}_{L_1(j)}(X))$\\
Return $(K,X,Y : \rho, \pi, F)$\smallskip
}{
\underline{$\procIcInv(K,X,Y : \rho, \pi, F)$}\\
$(i,j)\gets\rho(K)$\\
If $i=0$ then\\
\ind $Y \gets Y \xor F^{-1}_{L_0(j)}(X)$\\
Else\\
\ind $Y \gets Y \xor F_{L_1(j)}(\pi^{-1}(X))$\\
Return $(K,X,Y : \rho, \pi, F)$\smallskip
}
\vspace{-2ex}
\caption{Reduction adversary used in proof of \thref{thm:de-reduction}.}
\label{fig:de-games}
\hrulefill
\end{figure}

The adversary $\advA'$ is defined in \figref{fig:de-games}.  It
samples a permutation $\rho$ on $\bits^k$, a permutation $\pi$ on
$\bits^n$, and a cipher $F$.  The permutation $\rho$ is used to
provide a random map from the keys $K\in\bits^k$ to the elements of
the lists $L_0$ and $L_1$.  We will interpret $\rho(K)$ as a tuple
$(i,j)$ with $i\in\{0,1\}$ and $j\in[2^k/2]$.  Then $K$ gets mapped to
$L_i(j)$.  Therefore either none of the keys map to the same element
or a single pair of them maps to the same element.

Adversary $\advA$'s queries to $\procEv$ and $\procInv$ are answered using $\pi$ and $\pi^{-1}$.
Its queries to the ideal cipher are more complicated and are handled by the oracles $\procIc$ and $\procIcInv$.
We can verify that these oracles define permutations on bitstring inputs, so $\advA'$ is a well defined quantum adversary.
Consider a key $K$ and interpret $\rho(K)$ as a tuple $(i,j)$ as described above.
If $i=0$, then ideal cipher queries for it are answered as if $\ic_K(\cdot)=F_{L_0(j)}(\cdot)$.
If $i=1$, then ideal cipher queries for it are answered as if $\ic_K(\cdot)=\pi(F^{-1}_{L_1(j)}(\cdot))$.\footnote{When using output of $L$ as keys for $F$ we are identifying the elements of $[\rng]$ with elements of $\bits^{\lceil\lg\rng\rceil}$ in the standard manner.}
If the list element $K$ is not mapped to by any other keys, then the indexing into $F$ ensures that $\ic_K(\cdot)$ is independent of $\pi$ and $\ic_{K'}$ for all other $K'$.
If $K$ and $K'$ map to the same list element (and $i=0$ for $K$), then $\ic_K(\cdot)$ and $\ic_K'(\cdot)$ are random permutations conditioned on $\ic_K'(\ic_K(\cdot))=\pi(\cdot)$ and independent of all other $\ic_{K''}(\cdot)$.
	
In $\sH_0$, the permutation of $\procEv$ and $\procInv$ is independent of each $\ic_K(\cdot)$ which are themselves independent of each other.
So this perfectly matches the view presented to $\advA$ by $\advA'$ when the lists are disjoint.
In $\sH_1$, each $\ic_K(\cdot)$ is pairwise independent and the permutation of $\procEv$ and $\procInv$ is defined to equal $\EF_{K_2}(\EF_{K_1}(\cdot))$.
This perfectly matches the view presented to $\advA$ by $\advA'$ when the lists have one element in common because we can think of it as just having changed the order in which the permutations $\EF_{K_2}(\EF_{K_1}(\cdot))$, $\EF_{K_1}(\cdot)$, and $\EF_{K_2}(\cdot)$ were sampled.\qed
\end{proof}

\subsection{The Hardness of List Disjointness}\label{sec:hard-ld}
If $\advA$ is an algorithm making at most $q$ classical oracle queries,
then it is not hard to prove that $\Adv^{\ld}_{\dom,\rng}(\advA)\leq q/\dom$.
If, instead, $\advA$ makes as most $q$ quantum oracle queries, the
correct bound is less straightforward.
In this section, we will prove the following result.
\begin{theorem}\label{thm:ld-hard}
	If $\advA$ is a quantum algorithm making at most $q$ queries to its oracle and $\dom\geq 32$ is a power of 2, then
	\begin{displaymath}
		\Adv^{\ld}_{\dom,3\dom^2}(\advA)
		\leq
		11\sqrt[6]{(q\cdot \lg \dom\cdot \lg\lg \dom)^3/\dom^2}.
	\end{displaymath}
\end{theorem}
We restrict attention to the case that $\dom$ is a power of 2 only for notational simplicity in the proof.
Essentially the same bound for more general $\dom$ follows from the same techniques.

Ambanis's $\bigoh(N^{2/3})$ query algorithm for element distinctness~\cite{Amb07} can be used to solve list disjointness and hence shows this is tight (up to logarithmic factors) for attackers achieving constant advantage.
The sixth root degrades the quality of the bound for lower parameter regimes.
An interesting question we leave open is whether this could be proven without the sixth root or the logarithmic factors.

\heading{Proof sketch.}  The starting point for our reduction is that
$\Omega(N^{2/3})$ lower bounds are known both the search and decision
versions of $\ed$~\cite{AS04,Zhandry15}.\footnote{In the proof we
  actually work with the advantage upper bounds, rather than the
  corresponding query lower bounds.} 
By slightly modifying Zhandry's~\cite{Zhandry15} technique for proving this, 
we instead get a bound on the hardness of $\oed\srch$.
Next, a simple reduction (split the list in half at random) shows that $\ld\srch$ is as hard as $\oed\srch$.

Then a ``binary search'' style reduction shows that $\ld\dc$ is as hard as $\ld\srch$.
In the reduction, the $\ld\srch$ algorithm repeatedly splits its lists in half and uses the $\ld\dc$ algorithm to determine which pair of lists contains the non-disjoint entries.
However, we need our reduction to work by running the $\ld\dc$ algorithm on a particular fixed size of list (the particular size we showed $\ld\dc$ is hard for) rather than running it on numerous shrinking sizes.
We achieve this by padding the lists with random elements.
The choice of $\rng=3\dom^2$ was made so that with good probability these random elements do not overlap with the actual list.
This padding adds the $\lg D$ term to our bound. 

Finally a generic technique allows us to relate the hardness of $\ld$ and $\ld\dc$.
Given an algorithm with high $\ld$ advantage we can run it multiple times to get a precise estimate of how frequently it is outputting $1$ and use that to determine what we want to output.
This last step is the primary cause of the sixth root in our bound; it required running the $\ld$ algorithm on the order of $1/\delta^2$ times to get a precise enough estimate, where $\delta$ is the advantage of the $\ld$ algorithm.
This squaring of $\delta$ in the query complexity of our $\ld\dc$ algorithm (together with the fact that the query complexity is cubed in our $\oed\srch$ bound) ultimately causes the sixth root.

\ifnum\eprintversion=0
	\ifsubmission
		Our formalization of this proof is given in the full version (it is current in \apref{ap:ld-proof-lems}).
	\else
		Our formalization of this proof is given in the full version~\cite{fullversion}.
	\fi
	In particular, the proof primarily consists of four lemmas which give quantitative statements  capturing the claims:
	\begin{enumerate}
		\item $\oed\srch$ is hard.
		\item If $\oed\srch$ is hard, then $\ld\srch$ is hard.
		\item If $\ld\srch$ is hard, then $\ld\dc$ is hard.
		\item If $\ld\dc$ is hard then $\ld$ is hard.
	\end{enumerate}
	The final theorem follows by combining the quantitative claims.
\else
	\heading{Constituent lemmas.}
	In the rest of the section we will state and prove lemmas corresponding to each of the step of the proof described above.
	In \secref{app:ld-proof}
	we apply them one at a time to obtain the specific
	bound claimed by \thref{thm:ld-hard}.
	
	\begin{lemma}[$\oed\srch$ is hard]\label{lem:oeds-hard}
	  If $\advA_{\oed\srch}$ is a quantum algorithm for $\oed\srch$ making at most $q$
	  queries to its oracle and $\dom\geq 32$, then
	  $\Adv^{\ed\srch}_{\dom,3\dom^2}(\advA_{\oed\srch})\leq 9 \cdot (q+2)^3/\dom^2$.
	\end{lemma}

	\begin{proof}
	  In~\cite{Zhandry15}, Zhandry shows that no $Q$-query algorithm can
	  distinguish between a random function and a random injective
	  function with domain $[D]$ and codomain $[R]$ with advantage better
	  than $(\pi^2/3)Q^3/R$, as long as $ D \le R$. 
	  We could build a distinguisher $\advD$ that on input $L:[D]\to[R]$ runs $(x,y)\gets\advA[L]$.
	  If $(x,y)$ is a collision (i.e., $x\neq y$ and $L(x)=L(x)$), then $\advD$ outputs 1.
	  It checks this by making two additional $L$ queries, so $Q=q+2$.
	  Otherwise it outputs zero.
	  Clearly, $(x,y)$ cannot be a collision when $L$ is an injection so it will always output zero.
	  Hence, if $\mathsf{1coll}$ denotes the event that $L$ contains exactly one collision we obtain a bound of
	  \begin{equation}\label{eq:one-coll}
	  	\Pr[\mathsf{1coll}]\cdot \Adv^{\ed\srch}_{\dom,3\dom^2}(\advA_{\oed\srch})
	  	\leq
	  	(\pi^2/3)(q+2)^3/R.
	  \end{equation}
	  
	  It remains to lower bound $\Pr[\mathsf{1coll}]$.
	  Note there are $\binom{D}{2}$ possible pairs of inputs that could be collisions.
	  Each pair has a $1/R$ chance of colliding.
	  By a union bound, the probability any other input has the same output as the collision is at most $(D-2)/R$, so there is at least a $1-(D-2)/R$ probability this does not occur.
	  Given the above there are $(D-2)$ inputs sampled from $R-1$ possible values, so by a union bound the probability none of them collide is at least $1-\binom{D-2}{2}/(R-1)$.
	  This gives
	  \begin{displaymath}
	  	\Pr[\mathsf{1coll}]\geq\frac{D(D-1)}{2R}\cdot\left(1-\frac{D-2}{R}\right)\cdot\left(1-\frac{(D-2)(D-3)}{2(R-1)}\right).
	  \end{displaymath} 
	  Now setting $R=3D^2$ and applying simple bounds (e.g. $D-2<D-1<D$ and $(D-3)/(R-1)<(D-2)/R$) gives,
	  \begin{displaymath}
	  	\Pr[\mathsf{1coll}]
	  	\geq\frac{(1-1/D)}{6}\cdot\left(1-\frac{1}{3D}\right)\cdot\left(1-\frac{1}{6}\right)
	  \end{displaymath}
	  Plugging this lower bound into equation \ref{eq:one-coll}, re-arranging, applying the bound $D\geq 32$, and rounding up to the nearest whole number gives the claimed bound $\Adv^{\ed\srch}_{\dom,3\dom^2}(\advA_{\oed\srch})
	  	\leq 9 (q+2)^3/D^3$.\qed
	\end{proof}

	\begin{lemma}[$\oed\srch$ hard $\Rightarrow$ $\ld\srch$ hard]\label{lem:lds-hard}
		Let $\dom$ be even.
		If $\advA_{\ld\srch}$ is a quantum algorithm for $\ld\srch$ making at most $q$ queries to its oracle,  then there is an algorithm $\advA_{\oed\srch}$ (described in the proof) such that $\Adv^{\ld\srch}_{\dom,\rng}(\advA_{\ld\srch})\leq 2\Adv^{\oed\srch}_{\dom,\rng}(\advA_{\oed\srch})$.
		Algorithm $\advA_{\oed\srch}$ makes at most $q$ queries to its oracle.
	\end{lemma} 
	
	\begin{proof}
		On input a list $L:[\dom]\to[\rng]$, the algorithm $\advA_{\oed\srch}$ will pick a random permutation $\pi:[\dom]\to[\dom]$.
		It runs $\advA_{\ld\srch}[L\circ\pi]$ and then, on receiving output $(x,y)\in[\dom/2]^2$, returns $(\pi^{-1}(x),\pi^{-1}(y+\dom/2))$.
		The permutation $\pi$ serves the role of splitting $L$ into two sublists for $\advA_{\ld\srch}$ at random.
		As long as the original collision of $L$ doesn't end up being put into the same sublist (which has probability less than $1/2$), $\advA_{\ld\srch}$ will be run on a valid $\ld\srch$ instance and $\advA_{\oed\srch}$ will succeed whenever $\advA_{\ld\srch}$ does.
		The claim follows.\qed
	\end{proof}

		\begin{figure}[t]
		\ifnum\eprintversion=0
	        \renewcommand{\lwid}{0.6}
	        \renewcommand{\mwid}{}
	        \renewcommand{\rwid}{}
		\else
	        \renewcommand{\lwid}{0.42}
	        \renewcommand{\mwid}{}
	        \renewcommand{\rwid}{}
		\fi
		\oneCol{\lwid}
		{
		\underline{Algorithm $\advA_{\ld\srch}[L_0,L_1]$}\\
		$L'\getsr\Inj(D,R)$; $L'_0\concat L'_1\gets L$\\
		$i\gets 0$; $L_0^0\gets L_0$; $L_1^0\gets L_1$\\
		Repeat\\
		\ind $i\gets i+1$\\
		\ind $L_{0,l}\concat L_{0,r}\gets L_0^{i-1}$\\
		\ind $L_{1,l}\concat L_{1,r}\gets L_1^{i-1}$\\
		\ind $(j^\ast,k^\ast)\gets(r,r)$\\
		\ind For $(j,k)\in\sett{(l,l),(l,r),(r,l)}$ do\\
		\ind\ind\commentt{$f\square g(x) = f(x)$ for $x\in\Dom{f}$ else $g(x)$}
		\ind \ind $b\getsr\advA_{\ld\dc}[L_{0,j}\square L'_0,L_{1,k}\square L'_1]$\\
		\ind \ind If $b=1$ then $(j^\ast,k^\ast)\gets(j,k)$\\
		\ind $L^i_0\gets L_{0,j^\ast}$; $L^i_1\gets L_{1,k^\ast}$\\
		Until $|\Dom{L^i_0}|=|\Dom{L^i_1}|=1$\\
		Pick $(x,y)\in \Dom{L^i_0}\times\Dom{L^i_1}$\\
		Return $(x,y)$
		\smallskip
		}
		\vspace{-2ex}
		\caption{Reduction algorithm $\advA_{\ld\srch}$ for \lemref{lem:ldd-hard}. For notational convenience we write the two lists as separate input, rather than combined into a single list.}
		\label{fig:ldd-hard}
		\hrulefill
		\end{figure}

	\begin{lemma}[$\ld\srch$ hard $\Rightarrow$ $\ld\dc$ hard]\label{lem:ldd-hard}
		Let $\dom$ be a power of two.
		If $\advA_{\ld}$ is a quantum algorithm for $\ld$ making at most $q$ queries to its oracle, then there is an $\ld\srch$ algorithm $\advA_{\ld\srch}$ (described in the proof) such that 
		\begin{align*}
			\Adv^{\ld\srch}_{\dom,\rng}(\advA_{{\ld\srch}}) 
			\geq
			1 - D^2/R - 1.5(\lg D-2)(1-\Adv^{\ld\dc}_{\dom,\rng}(\advA_{\ld\dc})).
		\end{align*}
		Algorithm $\advA_{{\ld\srch}}$ makes at most $3q\lg D$ queries to its oracle.
	\end{lemma}

	This theorem's bound is vacuous if $\Adv^{\ld\dc}_{\dom,\rng}(\advA_{\ld\dc})$ is too small.
	When applying the result we will have obtained $\advA_{\ld\dc}$ by amplifying the advantage of another adversary to ensure it is sufficiently large.

	\begin{proof}
		The algorithm $\advA_{\ld\srch}$ is given in \figref{fig:ldd-hard}.
		The intuition behind this algorithm is as follows.
		It wants to use the decision algorithm $\advA_{\ld\dc}$ to perform a binary search to find the overlap between $L_0$ and $L_0$.
		It runs for $\lg D/2 - 1$ rounds.
		In the $i$-th round, it splits the current left list $L^{i-1}_0$ into two sublists $L_{0,l}$, $L_{0,r}$ and the current right lists $L^{i-1}_1$ into two sublists $L_{1,l}$, $L_{1,r}$.\footnote{In code, $f\concat g \gets h$ for $h$ with domain $\Dom{h}=\sett{n,n+1,\dots,m}$ defines $f$ to be the restriction of $h$ to domain $\Dom{f}=\sett{n,n+1,\dots,\lfloor(n+m)/2\rfloor}$ and $g$ to be the restriction of $h$ to domain $\Dom{g}=\Dom{h}\setminus\Dom{f}$.}
		If $L^{i-1}_0$ and $L^{i-1}_1$ have an element in common, then one of the pairs of sublists $L_{0,j}$ and $L_{1,k}$ for $j,k\in\sett{l,r}$ must have an element in common.
		The decision algorithm $\advA_{\ld\dc}$ is run on different pairs to determine which contains the overlap.
		We recurse with chosen pair, until we are left with lists that have singleton domains at which point we presume the entries therein give the overlap. 
		
		Because we need to run the decision algorithm $\advA_{\ld\dc}$ on fixed size inputs we pad the sublists to be of a fixed size using lists $L'_0$ and $L'_1$ (the two halves of an injection $L'$ that we sampled locally).
		As long as the image of $L'$ does not overlap with the images of $L_0$ and $L_1$, this padding does not introduce any additional elements repetitions between or within the list input to $\advA_{\ld\dc}$.
		By a union bound, overlaps occurs with probability at most $D^2/R$.
	
		Conditioned on such overlaps not occurring, $\advA_{\ld\srch}$ will output the correct result if $\advA'_{\ld\dc}$ always answers correctly.
		To bound the probability of $\advA'_{\ld\dc}$ erring we can
	        note that it is run $3\cdot(\lg D-2)$ times and, each time,
	        has at most a
	        $(1-\Adv^{\ld\dc}_{\dom,\rng}(\advA'_{\ld\dc}))/2$ chance of error and
	        apply a union bound.
	        
		Put together we have
		\begin{align*}
			\Adv^{\ld\srch}_{\dom,\rng}(\advA_{{\ld\srch}}) 
			\geq
			1 - D^2/R - 1.5(\lg D-2)(1-\Adv^{\ld\dc}_{\dom,\rng}(\advA_{\ld\dc}).
		\end{align*}
		That $\advA_{{\ld\srch}}$ makes at most $3q(\lg D-2)$ oracle queries is clear.\qed
	\end{proof}

	\begin{lemma}[$\pp\dc$ hard $\Rightarrow$ $\pp$ hard]\label{lem:dc-to-cry}
		Let $\pp$ be any problem.
		Suppose $\advA_{\pp}$ is a quantum algorithm for $\pp$ making at most $q$ queries to its oracle, with $\Adv^{\pp}_{\dom,\rng}(\advA_{\pp})=\delta>0$.
		Then for any $t\in\N$ there is an algorithm $\advA_{\pp\dc}$ (described in the proof) such that $\Adv^{\pp\dc}_{\dom,\rng}(\advA_{\pp\dc})>1-2/2^t$.
		Algorithm $\advA_{\pp\dc}$ makes $q\cdot\lceil 4.5(t+1)\ln 2/\delta^2 \rceil$ queries to its oracle.
	\end{lemma}

	\begin{proof}
	For compactness, let $p^1=P_{\dom,\rng}^1(\advA_{\pp})$ denote the minimum probability that $\advA_{\pp}$ outputs $1$ on instances in the language, $p^0=P_{\dom,\rng}^0(\advA_{\pp})$ denote the maximum probability that $\advA_{\pp}$ outputs $1$ on instances not in the language, and $\delta=\Adv^{\pp}_{\dom,\rng}(\advA_{\pp})=p^1-p^0$.
	We define $\advA_{\pp\dc}$ to be the algorithm that, on input $L$, runs $n$ independent copies of $\advA_{\pp}[L]$ (with $n = \lceil 4.5(t+1)*\ln 2/\delta^2 \rceil$) and calculates the average $p$ of all the values output by $\advA_{\pp}[L]$. (Think of this as an estimate of the probability $\advA_{\pp}$ outputs 1 on this input.) If $p < p^0 + \delta/2$, $\advA_{\pp\dc}$ outputs 0, otherwise it outputs 1.
	
	Let $X_i$ denote the output of the $i$-th execution of $\advA_{\pp\dc}$.
	An inequality of Hoeffding~\cite{hoeffding} bounds how far the average of independent random variables $0\leq X_i\leq 1$ can differ from the expectation by 
	\begin{align*}
		\Pr[\left|\sum_{i=1}^n X_i/n - \Exp{\sum_{i=1}^n X_i/n}\right| > \epsilon]<2e^{-2\epsilon^2n}
	\end{align*}
	for any $\epsilon>0$.
	Let $p'$ denote the expected value of $p$ (which is also the expected value of $X_i$). 
	If $L$ is in the language, then $p^1\leq p$.
	Otherwise $p^0\geq p$.
	In either case, $\advA_{\pp\dc}$ will output the correct answer if $p$ does not differ from $p'$ by more than $\delta/3$ (because this is strictly less than $\delta/2$).
	Applying the above inequality tells us that $\Pr[|p - p'|>\delta/3]<2e^{-2\delta^2n/9}\leq 2^{-t}$.
	(The value of $n$ was chosen to make this last inequality hold.)
	
	Hence $P_{\dom,\rng}^1(\advA_{\pp\dc})>1-2^{-t}$, $P_{\dom,\rng}^0(\advA_{\pp\dc})<2^{-t}$, and $\Adv^{\pp\dc}_{\dom,\rng}(\advA_{\pp\dc})>1-2/2^{t}$.
	The bound on $\advA_{\pp\dc}$'s number of queries is clear.\qed
	\end{proof}
\fi

\ifnum\eprintversion=0
\else
\ifnum\eprintversion=0
	\section{Proof of \thref{thm:ld-hard}}\label{app:ld-proof}
\else
	\subsection{Proof of \thref{thm:ld-hard}}\label{app:ld-proof}
\fi
Let $\advA_{\ld}$ be our given list disjointness adversary which makes $q$ oracle queries and has advantage $\delta=\Adv^{\ld}_{\dom,\rng}(\advA_{\ld})>0$.
If we apply \lemref{lem:dc-to-cry} with $t=\lg(12\lg D)$, then we get an adversary $\advA_{\ld\dc}$ which makes $q_{\ld\dc}=q\cdot\lceil 4.5(t+1)\ln 2/\delta^2 \rceil < (10q\lg\lg D)/\delta^2$ oracle queries. (We used here the assumption $D\geq 32$ to simplify constants.) This adversary has advantage
	\begin{align*}
		\Adv^{\ld\dc}_{\dom,3\dom^2}(\advA_{\ld\dc})>1-2/2^t=1-1/(6\lg(D)).
	\end{align*}
Next applying \lemref{lem:ldd-hard}, gives us $\advA_{{\ld\srch}}$ which makes fewer than $q_{\ld\srch}=30q(\lg D)(\lg\lg D)/\delta^2$ oracle queries and has advantage 
	\begin{align*}
		\Adv^{\ld\srch}_{\dom,\rng}(\advA_{{\ld\srch}}) 
		&\geq
		1 - D^2/R - 1.5(\lg D-2)(1-\Adv^{\ld\dc}_{\dom,\rng}(\advA_{\ld\dc}))\\
		&> 1 - 1/3 - 1.5(\lg D-2)(6\lg D)^{-1} > 1-1/3-1/4 = 5/12.
	\end{align*}
Lemmas \ref{lem:oeds-hard} and \ref{lem:lds-hard} together bound this advantage from the other direction.
In particular, we get that
	\begin{align*}
		5/12
		<
		\Adv^{\ld\srch}_{\dom,\rng}(\advA_{{\ld\srch}})
		\leq
		18 \cdot (q_{\ld\srch}+2)^3/\dom^2.
	\end{align*}
Using the assumption $D\geq 32$ we can bound $q_{\ld\srch}+2$ by $31q(\lg D)(\lg\lg D)/\delta^2$.
Plugging this in and solving for $\delta$ gives our claimed bound of
	\begin{align*}
		\delta < 11\sqrt[6]{(q\cdot \lg D\cdot \lg\lg D)^3/D^2}.	
	\end{align*}

\fi

\section*{Acknowledgements}
\label{sec:acks}

Joseph Jaeger and Stefano Tessaro were partially supported by NSF grants CNS-1930117 (CAREER),
CNS-1926324, CNS-2026774, a Sloan Research Fellowship, and a JP Morgan
Faculty Award. 
Joseph Jaeger's work done while at the University of Washington.
Fang Song thanks Robin Kothari for helpful discussion on the
element distinctness problem. Fang Song was supported by NSF grants
CCF-2041841, CCF-2042414, and CCF-2054758 (CAREER).


\appendix
\addcontentsline{toc}{section}{Appendices}
\ifnum\eprintversion=0
	\ifsubmission

\section{Probabilistic Analysis for Proof of \thref{thm:fx-kpa}}\label{app:na-proof-extra}
In this appendix we give a detailed probability analysis of Claim (i) from the proof of \thref{thm:fx-kpa}.

The view of $\advA$ when run by $\advA'$ in $\sG^{\dist}_{\distD,1}$ is determined by $f_1$, $f^{-1}_1$, and $z=(T,T^{-1})$. These are chosen such that $T[M_i]=f_1(K_1,M_i\xor K_2)\xor K_2$ for $i=1,\dots,q'$ and $T^{-1}[Y_i]=f^{-1}_1(K_1,Y_i\xor K_2)\xor K_2$ for $i=q'+1,\dots,q$. Consequently to ensure this matches the view from $\sG^{\sprp}_{\FF,1}$ we need to show that $\distD$ choses $(f_1,K_1,K_2)$ uniformly from $\Icm(k,n)\times\bits^k\times\bits^n$. It is clear that $K_1$ and $K_2$ are uniform. Furthermore, $f_1(K,\cdot)=f_0(K,\cdot)$ for $K\neq K_1$ so these are sampled correctly. Hence we can think of these values as fixed and argue that the distribution induced over $f_1(K_1,\cdot)$ is uniform over all permutations on $\bits^n$.

\newcommand{\ff}{\mathbf{f}}
Let $\ff\in\Icm(k,n)$, $N=2^n$, and $E$ denote the event that $f_1(K_1,\cdot)=\ff(\cdot)$.
Let $E_1$ denote the event that $T[M_i]=\ff(M_i\xor K_2)\xor K_2$ for $i=1,\dots,q$ after Step 1. 
The second for loop of Step 3 programs $f_1$ to satisfy $f_1(K_1,M_i\xor K_2)=T[M_i]\xor K_2$ so $E_1$ is a necessary condition for $E$.
Note that $E_1$ requires $Q$ values to have been sampled correctly in Step 1's for loops where $Q=|\sett{\ff(M_i\xor K_2)\xor K_2:1\leq i \leq q'}\cup \sett{\ff(Y_i\xor K_2)\xor K_2:q'\leq i \leq q}|$.
Hence,
$$\Pr[E]= \Pr[E|E_1]\cdot\frac{(N-Q)!}{N!}.$$
Let $E_2$ denote the event that Step 2 samples an $f_0$ which is consistent with $\ff$.
That is to say, that $f_0$ is chosen such that $\ff(\setI\cup\setI')=\setO\cup\setO'$ and $\ff(x)=f_0(K_1,x)$ for $x\not\in\setI\cup\setI'$.
The first for loop in Step 3 programs $f_1(K_1,x)=f_0(K_1,x)$ for $x\not\in\setI\cup\setI'$.
The second and third for loop in Step 3 programs $f_1(K_1,\setI\cup\setI')$ to have values in $\setO\cup\setO'$ so $E_2$ is a necessary condition for $E$.
Let $M(f_0)=|\setI\setminus\setI'|=|\setO\setminus \setO'|$ and let $E^m_2$ denote the event that $M(f_0)=m$.
$$\Pr[E|E_1]=\sum_m \Pr[E^m_2|E_1]\cdot\Pr[E_2|E_1,E^m_2]\cdot\Pr[E|E_1,E^m_2,E_2].$$
We can think of Step 2 as lazily sampling $f_0(K_1,\cdot)$ by first sampling $f_0(K_1,x)$ for $x\in\setI$, then sampling $f^{-1}_0(K_1,y)$ for $y\in\setO\setminus \setO'$, and then  sampling $f_0(K_1,x)$ for $x\not\in\setI\cup\setI'$.
The event $E_2$ requires that the $m$ values sampled in this second step are the elements of $\ff^{-1}(\setO'\setminus\setO)$ and that on the third step $\ff(x)$ is sampled for $f_0(K_1,x)$ for each $x\not\in\setI\cup\setI'$. Hence,
$$\Pr[E_2|E_1,E^m_2]=\frac{1}{\binom{N-Q}{m}}\cdot\frac{1}{(N-Q-m)!}.$$
For $E$ to occur (conditioned on $E_1$, $E^m_2$, and $E_2$) the third for loop in Step 3 must sample $m$ values correctly, so $\Pr[E|E_1,E^m_2,E_2]=1/m!$.
Putting everything together we have
\begin{align*}
	\Pr[E]=\frac{(N-Q)!}{N!}\sum_m \Pr[E^m_2|E_1]\cdot\frac{m!\cdot(N-Q-m)!}{(N-Q)!}\cdot\frac{1}{(N-Q-m)!}\cdot\frac{1}{m!}=\frac{1}{N!}.
\end{align*}
So $f_1$ is uniformly distributed, as desired.


\ifsubmission

\section{Lemmas for \thref{thm:ld-hard}}\label{ap:ld-proof-lems}
In this section we will state and prove lemmas corresponding to each of the step of the proof of \thref{thm:ld-hard}.
In \apref{app:ld-proof} we apply them one at a time to obtain the specific
bound claimed by \thref{thm:ld-hard}.

\begin{lemma}[$\oed\srch$ is hard]\label{lem:oeds-hard}
  If $\advA_{\oed\srch}$ is a quantum algorithm for $\oed\srch$ making at most $q$
  queries to its oracle and $\dom\geq 32$, then
  $\Adv^{\ed\srch}_{\dom,3\dom^2}(\advA_{\oed\srch})\leq 9 \cdot (q+2)^3/\dom^2$.
\end{lemma}

\begin{proof}
  In~\cite{Zhandry15}, Zhandry shows that no $Q$-query algorithm can
  distinguish between a random function and a random injective
  function with domain $[D]$ and codomain $[R]$ with advantage better
  than $(\pi^2/3)Q^3/R$, as long as $ D \le R$. 
  We could build a distinguisher $\advD$ that on input $L:[D]\to[R]$ runs $(x,y)\gets\advA[L]$.
  If $(x,y)$ is a collision (i.e., $x\neq y$ and $L(x)=L(x)$), then $\advD$ outputs 1.
  It checks this by making two additional $L$ queries, so $Q=q+2$.
  Otherwise it outputs zero.
  Clearly, $(x,y)$ cannot be a collision when $L$ is an injection so it will always output zero.
  Hence, if $\mathsf{1coll}$ denotes the event that $L$ contains exactly one collision we obtain a bound of
  \begin{equation}\label{eq:one-coll}
  	\Pr[\mathsf{1coll}]\cdot \Adv^{\ed\srch}_{\dom,3\dom^2}(\advA_{\oed\srch})
  	\leq
  	(\pi^2/3)(q+2)^3/R.
  \end{equation}
  
  It remains to lower bound $\Pr[\mathsf{1coll}]$.
  Note there are $\binom{D}{2}$ possible pairs of inputs that could be collisions.
  Each pair has a $1/R$ chance of colliding.
  By a union bound, the probability any other input has the same output as the collision is at most $(D-2)/R$, so there is at least a $1-(D-2)/R$ probability this does not occur.
  Given the above there are $(D-2)$ inputs sampled from $R-1$ possible values, so by a union bound the probability none of them collide is at least $1-\binom{D-2}{2}/(R-1)$.
  This gives
  \begin{displaymath}
  	\Pr[\mathsf{1coll}]\geq\frac{D(D-1)}{2R}\cdot\left(1-\frac{D-2}{R}\right)\cdot\left(1-\frac{(D-2)(D-3)}{2(R-1)}\right).
  \end{displaymath} 
  Now setting $R=3D^2$ and applying simple bounds (e.g. $D-2<D-1<D$ and $(D-3)/(R-1)<(D-2)/R$) gives,
  \begin{displaymath}
  	\Pr[\mathsf{1coll}]
  	\geq\frac{(1-1/D)}{6}\cdot\left(1-\frac{1}{3D}\right)\cdot\left(1-\frac{1}{6}\right)
  \end{displaymath}
  Plugging this lower bound into equation \ref{eq:one-coll}, re-arranging, applying the bound $D\geq 32$, and rounding up to the nearest whole number gives the claimed bound $\Adv^{\ed\srch}_{\dom,3\dom^2}(\advA_{\oed\srch})
  	\leq 9 (q+2)^3/D^3$.\qed
\end{proof}

\begin{lemma}[$\oed\srch$ hard $\Rightarrow$ $\ld\srch$ hard]\label{lem:lds-hard}
	Let $\dom$ be even.
	If $\advA_{\ld\srch}$ is a quantum algorithm for $\ld\srch$ making at most $q$ queries to its oracle,  then there is an algorithm $\advA_{\oed\srch}$ (described in the proof) such that $\Adv^{\ld\srch}_{\dom,\rng}(\advA_{\ld\srch})\leq 2\Adv^{\oed\srch}_{\dom,\rng}(\advA_{\oed\srch})$.
	Algorithm $\advA_{\oed\srch}$ makes at most $q$ queries to its oracle.
\end{lemma} 

\begin{proof}
	On input a list $L:[\dom]\to[\rng]$, the algorithm $\advA_{\oed\srch}$ will pick a random permutation $\pi:[\dom]\to[\dom]$.
	It runs $\advA_{\ld\srch}[L\circ\pi]$ and then, on receiving output $(x,y)\in[\dom/2]^2$, returns $(\pi^{-1}(x),\pi^{-1}(y+\dom/2))$.
	The permutation $\pi$ serves the role of splitting $L$ into two sublists for $\advA_{\ld\srch}$ at random.
	As long as the original collision of $L$ doesn't end up being put into the same sublist (which has probability less than $1/2$), $\advA_{\ld\srch}$ will be run on a valid $\ld\srch$ instance and $\advA_{\oed\srch}$ will succeed whenever $\advA_{\ld\srch}$ does.
	The claim follows.\qed
\end{proof}

	\begin{figure}[t]
	\ifnum\eprintversion=0
        \renewcommand{\lwid}{0.6}
        \renewcommand{\mwid}{}
        \renewcommand{\rwid}{}
	\else
        \renewcommand{\lwid}{0.42}
        \renewcommand{\mwid}{}
        \renewcommand{\rwid}{}
	\fi
	\oneCol{\lwid}
	{
	\underline{Algorithm $\advA_{\ld\srch}[L_0,L_1]$}\\
	$L'\getsr\Inj(D,R)$; $L'_0\concat L'_1\gets L$\\
	$i\gets 0$; $L_0^0\gets L_0$; $L_1^0\gets L_1$\\
	Repeat\\
	\ind $i\gets i+1$\\
	\ind $L_{0,l}\concat L_{0,r}\gets L_0^{i-1}$\\
	\ind $L_{1,l}\concat L_{1,r}\gets L_1^{i-1}$\\
	\ind $(j^\ast,k^\ast)\gets(r,r)$\\
	\ind For $(j,k)\in\sett{(l,l),(l,r),(r,l)}$ do\\
	\ind\ind\commentt{$f\square g(x) = f(x)$ for $x\in\Dom{f}$ else $g(x)$}
	\ind \ind $b\getsr\advA_{\ld\dc}[L_{0,j}\square L'_0,L_{1,k}\square L'_1]$\\
	\ind \ind If $b=1$ then $(j^\ast,k^\ast)\gets(j,k)$\\
	\ind $L^i_0\gets L_{0,j^\ast}$; $L^i_1\gets L_{1,k^\ast}$\\
	Until $|\Dom{L^i_0}|=|\Dom{L^i_1}|=1$\\
	Pick $(x,y)\in \Dom{L^i_0}\times\Dom{L^i_1}$\\
	Return $(x,y)$
	\smallskip
	}
	\vspace{-2ex}
	\caption{Reduction algorithm $\advA_{\ld\srch}$ for \lemref{lem:ldd-hard}. For notational convenience we write the two lists as separate input, rather than combined into a single list.}
	\label{fig:ldd-hard}
	\hrulefill
	\end{figure}

\begin{lemma}[$\ld\srch$ hard $\Rightarrow$ $\ld\dc$ hard]\label{lem:ldd-hard}
	Let $\dom$ be a power of two.
	If $\advA_{\ld}$ is a quantum algorithm for $\ld$ making at most $q$ queries to its oracle, then there is an $\ld\srch$ algorithm $\advA_{\ld\srch}$ (described in the proof) such that 
	\begin{align*}
		\Adv^{\ld\srch}_{\dom,\rng}(\advA_{{\ld\srch}}) 
		\geq
		1 - D^2/R - 1.5(\lg D-2)(1-\Adv^{\ld\dc}_{\dom,\rng}(\advA_{\ld\dc})).
	\end{align*}
	Algorithm $\advA_{{\ld\srch}}$ makes at most $3q\lg D$ queries to its oracle.
\end{lemma}

This theorem's bound is vacuous if $\Adv^{\ld\dc}_{\dom,\rng}(\advA_{\ld\dc})$ is too small.
When applying the result we will have obtained $\advA_{\ld\dc}$ by amplifying the advantage of another adversary to ensure it is sufficiently large.

\begin{proof}
	The algorithm $\advA_{\ld\srch}$ is given in \figref{fig:ldd-hard}.
	The intuition behind this algorithm is as follows.
	It wants to use the decision algorithm $\advA_{\ld\dc}$ to perform a binary search to find the overlap between $L_0$ and $L_0$.
	It runs for $\lg D/2 - 1$ rounds.
	In the $i$-th round, it splits the current left list $L^{i-1}_0$ into two sublists $L_{0,l}$, $L_{0,r}$ and the current right lists $L^{i-1}_1$ into two sublists $L_{1,l}$, $L_{1,r}$.\footnote{In code, $f\concat g \gets h$ for $h$ with domain $\Dom{h}=\sett{n,n+1,\dots,m}$ defines $f$ to be the restriction of $h$ to domain $\Dom{f}=\sett{n,n+1,\dots,\lfloor(n+m)/2\rfloor}$ and $g$ to be the restriction of $h$ to domain $\Dom{g}=\Dom{h}\setminus\Dom{f}$.}
	If $L^{i-1}_0$ and $L^{i-1}_1$ have an element in common, then on of the pairs of sublists $L_{0,j}$ and $L_{1,k}$ for $j,k\in\sett{l,r}$ must have an element in common.
	The decision algorithm $\advA_{\ld\dc}$ is run on different pairs to determine which contains the overlap.
	We recurse with chosen pair, until we are left with lists that have singleton domains at which point we presume the entries therein give the overlap. 
	
	Because we need to run the decision algorithm $\advA_{\ld\dc}$ on fixed size inputs we pad the sublists to be of a fixed size using lists $L'_0$ and $L'_1$ (the two halves of an injection $L'$ that we sampled locally).
	As long as the image of $L'$ does not overlap with the images of $L_0$ and $L_1$, this padding does not introduce any additional elements repetitions between or within the list input to $\advA_{\ld\dc}$.
	By a union bound, overlaps occurs with probability at most $D^2/R$.

	Conditioned on such overlaps not occurring, $\advA_{\ld\srch}$ will output the correct result if $\advA'_{\ld\dc}$ always answers correctly.
	To bound the probability of $\advA'_{\ld\dc}$ erring we can
        note that it is run $3\cdot(\lg D-2)$ times and, each time,
        has at most a
        $(1-\Adv^{\ld\dc}_{\dom,\rng}(\advA'_{\ld\dc}))/2$ chance of error and
        apply a union bound.
        
	Put together we have
	\begin{align*}
		\Adv^{\ld\srch}_{\dom,\rng}(\advA_{{\ld\srch}}) 
		\geq
		1 - D^2/R - 1.5(\lg D-2)(1-\Adv^{\ld\dc}_{\dom,\rng}(\advA_{\ld\dc}).
	\end{align*}
	That $\advA_{{\ld\srch}}$ makes at most $3q(\lg D-2)$ oracle queries is clear.\qed
\end{proof}

\begin{lemma}[$\pp\dc$ hard $\Rightarrow$ $\pp$ hard]\label{lem:dc-to-cry}
	Let $\pp$ be any problem.
	Suppose $\advA_{\pp}$ is a quantum algorithm for $\pp$ making at most $q$ queries to its oracle, with $\Adv^{\pp}_{\dom,\rng}(\advA_{\pp})=\delta>0$.
	Then for any $t\in\N$ there is an algorithm $\advA_{\pp\dc}$ (described in the proof) such that $\Adv^{\pp\dc}_{\dom,\rng}(\advA_{\pp\dc})>1-2/2^t$.
	Algorithm $\advA_{\pp\dc}$ makes $q\cdot\lceil 4.5(t+1)\ln 2/\delta^2 \rceil$ queries to its oracle.
\end{lemma}

\begin{proof}
For compactness, let $p^1=P_{\dom,\rng}^1(\advA_{\pp})$ denote the minimum probability that $\advA_{\pp}$ outputs $1$ on instances in the language, $p^0=P_{\dom,\rng}^0(\advA_{\pp})$ denote the maximum probability that $\advA_{\pp}$ outputs $1$ on instances not in the language, and $\delta=\Adv^{\pp}_{\dom,\rng}(\advA_{\pp})=p^1-p^0$.
We define $\advA_{\pp\dc}$ to be the algorithm that, on input $L$, runs $n$ independent copies of $\advA_{\pp}[L]$ (with $n = \lceil 4.5(t+1)*\ln 2/\delta^2 \rceil$) and calculates the average $p$ of all the values output by $\advA_{\pp}[L]$. (Think of this as an estimate of the probability $\advA_{\pp}$ outputs 1 on this input.) If $p < p^0 + \delta/2$, $\advA_{\pp\dc}$ outputs 0, otherwise it outputs 1.

Let $X_i$ denote the output of the $i$-th execution of $\advA_{\pp\dc}$.
An inequality of Hoeffding~\cite{hoeffding} bounds how far the average of independent random variables $0\leq X_i\leq 1$ can differ from the expectation by 
\begin{align*}
	\Pr[\left|\sum_{i=1}^n X_i/n - \Exp{\sum_{i=1}^n X_i/n}\right| > \epsilon]<2e^{-2\epsilon^2n}
\end{align*}
for any $\epsilon>0$.
Let $p'$ denote the expected value of $p$ (which is also the expected value of $X_i$). 
If $L$ is in the language, then $p^1\leq p$.
Otherwise $p^0\geq p$.
In either case, $\advA_{\pp\dc}$ will output the correct answer if $p$ does not differ from $p'$ by more than $\delta/3$ (because this is strictly less than $\delta/2$).
Applying the above inequality tells us that $\Pr[|p - p'|>\delta/3]<2e^{-2\delta^2n/9}\leq 2^{-t}$.
(The value of $n$ was chosen to make this last inequality hold.)

Hence $P_{\dom,\rng}^1(\advA_{\pp\dc})>1-2^{-t}$, $P_{\dom,\rng}^0(\advA_{\pp\dc})<2^{-t}$, and $\Adv^{\pp\dc}_{\dom,\rng}(\advA_{\pp\dc})>1-2/2^{t}$.
The bound on $\advA_{\pp\dc}$'s number of queries is clear.\qed
\end{proof}
\fi

\ifnum\eprintversion=0
	\section{Proof of \thref{thm:ld-hard}}\label{app:ld-proof}
\else
	\subsection{Proof of \thref{thm:ld-hard}}\label{app:ld-proof}
\fi
Let $\advA_{\ld}$ be our given list disjointness adversary which makes $q$ oracle queries and has advantage $\delta=\Adv^{\ld}_{\dom,\rng}(\advA_{\ld})>0$.
If we apply \lemref{lem:dc-to-cry} with $t=\lg(12\lg D)$, then we get an adversary $\advA_{\ld\dc}$ which makes $q_{\ld\dc}=q\cdot\lceil 4.5(t+1)\ln 2/\delta^2 \rceil < (10q\lg\lg D)/\delta^2$ oracle queries. (We used here the assumption $D\geq 32$ to simplify constants.) This adversary has advantage
	\begin{align*}
		\Adv^{\ld\dc}_{\dom,3\dom^2}(\advA_{\ld\dc})>1-2/2^t=1-1/(6\lg(D)).
	\end{align*}
Next applying \lemref{lem:ldd-hard}, gives us $\advA_{{\ld\srch}}$ which makes fewer than $q_{\ld\srch}=30q(\lg D)(\lg\lg D)/\delta^2$ oracle queries and has advantage 
	\begin{align*}
		\Adv^{\ld\srch}_{\dom,\rng}(\advA_{{\ld\srch}}) 
		&\geq
		1 - D^2/R - 1.5(\lg D-2)(1-\Adv^{\ld\dc}_{\dom,\rng}(\advA_{\ld\dc}))\\
		&> 1 - 1/3 - 1.5(\lg D-2)(6\lg D)^{-1} > 1-1/3-1/4 = 5/12.
	\end{align*}
Lemmas \ref{lem:oeds-hard} and \ref{lem:lds-hard} together bound this advantage from the other direction.
In particular, we get that
	\begin{align*}
		5/12
		<
		\Adv^{\ld\srch}_{\dom,\rng}(\advA_{{\ld\srch}})
		\leq
		18 \cdot (q_{\ld\srch}+2)^3/\dom^2.
	\end{align*}
Using the assumption $D\geq 32$ we can bound $q_{\ld\srch}+2$ by $31q(\lg D)(\lg\lg D)/\delta^2$.
Plugging this in and solving for $\delta$ gives our claimed bound of
	\begin{align*}
		\delta < 11\sqrt[6]{(q\cdot \lg D\cdot \lg\lg D)^3/D^2}.	
	\end{align*}

	\fi
\else
	
\fi

\end{document}